%% file: main.tex
\documentclass[a4paper,UKenglish,cleveref, autoref, thm-restate]{lipics-v2021}
\title{Analysis of logics with arithmetic} 

\author{Michael Benedikt}{University of Oxford, UK }{}{}{}
\author{Chia-Hsuan Lu} {University of Oxford, UK}{}{}{}
\author{Tony Tan} {University of Liverpool, UK}{}{}{}

\nolinenumbers

\input{macros}

\authorrunning{Benedikt, Lu, and Tan} 

\Copyright{M. Benedikt, C-H. Lu, and T. Tan} 



\category{} 

\relatedversion{}

\EventEditors{}
\EventNoEds{1}
\EventLongTitle{}
\EventShortTitle{}
\EventAcronym{}
\EventYear{}
\EventDate{}
\EventLocation{}
\EventLogo{}
\SeriesVolume{}

\begin{document}

\maketitle
\input{0-abstract}

\input{1-intro}
\input{2-prelims}

\input{2a-kleenestar}

\input{2b-prelims-type}

\input{3-logics}

\input{4-gptwo}

\input{5-ctwo}

\input{6-gptwospectrum}

\input{7-related}

\input{8-conc}

\bibliography{references}

\input{appendix}

\end{document}

%% file: macros.tex
\makeatletter
\newcommand\notsotiny{\@setfontsize\notsotiny{8}{7}}
\makeatother

\usepackage[normalem]{ulem}
\usepackage{enumerate}
\usepackage{paralist}
\usepackage{tikz,tkz-graph,tkz-berge}
\usepackage{algorithm}
\usepackage{algpseudocode}

\usepackage{amsthm}

\usepackage{amssymb}
\usepackage{stmaryrd}

\usepackage{tabularray}

\usepackage{mathtools}

\usepackage{mathtools}
\newcommand{\transs}[1] {\xrightarrow{ {#1} }}

\newcommand{\myparagraph}[1]{\medskip\noindent \textbf{#1.}}

\newcommand{\cA}{\mathcal{A}}
\newcommand{\cB}{\mathcal{B}}
\newcommand{\cC}{\mathcal{C}}
\newcommand{\cCtilde}{\tilde{\mathcal{C}}}

\newcommand{\cD}{\mathcal{D}}

\newcommand{\cG}{\mathcal{G}}
\newcommand{\cH}{\mathcal{H}}

\newcommand{\cK}{\mathsf{TwoTps}}

\newcommand{\cM}{\mathcal{M}}

\newcommand{\cO}{\mathcal{O}}

\newcommand{\cQ}{\mathcal{Q}}
\newcommand{\cQtilde}{\tilde{\mathcal{Q}}}

\newcommand{\cS}{\mathcal{S}}

\newcommand{\bbN}{\mathbb{N}}
\newcommand{\bbB}{\mathbb{B}}

\newcommand{\bbZ}{\mathbb{Z}}

\newcommand{\bfA}{\mathbf{A}}
\newcommand{\bfAtilde}{\tilde{\mathbf{A}}}
\newcommand{\bfB}{\mathbf{B}}

\newcommand{\bfI}{\mathbf{I}}

\newcommand{\bfa}{\mathbf{a}}
\newcommand{\bfb}{\mathbf{b}}
\newcommand{\bfc}{\mathbf{c}}
\newcommand{\bfe}{\mathbf{e}}
\newcommand{\bff}{\mathbf{f}}
\newcommand{\bfg}{\mathbf{g}}

\newcommand{\bfk}{\mathbf{k}}

\newcommand{\bfu}{\mathbf{u}}
\newcommand{\bfv}{\mathbf{v}}

\newcommand{\bfx}{\mathbf{x}}
\newcommand{\bfy}{\mathbf{y}}

\newcommand{\bfzero}{\mathbf{0}}
\newcommand{\bfone}{\mathbf{1}}

\newcommand{\vocab}{\sigma}

\newcommand{\nexp}  {\mathsf{NEXP}}
\newcommand{\mexp}  {\mathsf{EXP}}
\newcommand{\expp}  {\mathsf{EXP}}
\newcommand{\twoexpp}  {\mathsf{2EXP}}
\newcommand{\exptime}{\mexp}
\newcommand{\nexptime}  {\mathsf{NEXP}}

\newcommand{\threenexp}  {\mathsf{3NEXP}}

\newcommand{\Ctwo}    {\mathsf{C^2}}
\newcommand{\Ctwoplusglobal} { \Ctwo^{+g}}

\newcommand{\onetypeof}{\mathsf{OneTp}}
\newcommand{\twotypeof}{\mathsf{TwoTp}}
\newcommand{\twotypes}{\mathsf{TwoTps}}

\newcommand{\onetypes}{\mathsf{OneTps}}
\newcommand{\Onetypes}{\mathsf{OneTps}}
\newcommand{\Ptwo}    {\mathsf{P^2}}

\newcommand{\GCtwo}   {\mathsf{GC^2}}
\newcommand{\GPtwo}   {\mathsf{GP^2}}

\newcommand{\presby}[1]  {{\#_y\left[ #1 \right]}}

\newcommand{\set}[1]         {\left\{ #1 \right\}}
\newcommand{\setc}[2]        {\left\{ #1 \middle|\  #2 \right\}}
\newcommand{\tuple}[1]       {\left\langle #1 \right\rangle}
\newcommand{\intsinterval}[1]{{\left[ #1 \right]}} 
\newcommand{\intinterval}[2] {{\left[ #1 ,\ #2 \right]}}

\newcommand{\abs}[1]         {{\left| #1 \right|}}
\newcommand{\norm}[1]        {\left\Vert #1 \right\Vert}

\newcommand{\projp}[2]     {\mathsf{proj}_{ #1 }\left( #2 \right)}

\newcommand{\dual}[1] {\widetilde{ #1 }}
\newcommand{\etanull} {\eta_0}
\newcommand{\bh}[1] {\mathsf{bh}\left( #1 \right)}

\newcommand{\kstar}[1] {\mathsf{KleeneStar}\left( #1 \right)}

\newcommand{\parikh}[1] {\mathsf{Parikh}\left( #1 \right)}
\newcommand{\mpath}[1] {\mathsf{Path}\left( #1 \right)}
\newcommand{\sol}[1] {\left\llbracket #1 \right\rrbracket}
\newcommand{\solmin}[1] {\left\llbracket #1 \right\rrbracket_\text{min}}
\newcommand{\vle} {\preccurlyeq}

\newcommand{\unit}[1] {\bfe_{ #1 }}
\newcommand{\bigO}[1] {{\cO\left( #1 \right)}}
\newcommand{\maxlambda} {\widetilde{\lambda}}

\newcommand{\trif} {\triangleright}
\newcommand{\trib} {\triangleleft}

\newcommand{\cKvpnull}  {\cK^{\varphi, \mathsf{silent}}}

\newcommand{\cKvpf}     {\cK^{\varphi, \trif}}
\newcommand{\cKvpb}     {\cK^{\varphi, \trib}}
\newcommand{\cKvpfb}    {\cK^{\varphi, \trif\trib}}
\newcommand{\spec} {\mathsf{SPEC}}
\newcommand{\pispec} {\Pi\text{-}\mathsf{SPEC}}
\newcommand{\summ}  {\mathsf{SUM}}
\newcommand{\comp}  {\mathsf{COMP}}
\newcommand{\silent}  {\mathsf{SILENT}}
\newcommand{\bige}  {\mathsf{BIG}}
\newcommand{\match} {\mathsf{MATCH}}


%% file: 0-abstract.tex
\begin{abstract}
We present new results on finite satisfiability of logics with counting and arithmetic.
One result is a tight bound on the complexity of satisfiability of logics with so-called local Presburger quantifiers, which sum over neighbors of a node in a graph. A second contribution concerns computing a semilinear representation of the cardinalities associated with a formula in two variable logic extended with counting quantifiers. Such a representation allows you to get bounds not only on satisfiability for these logics, but for satisfiability in the presence of additional ``global cardinality constraints'': restrictions on cardinalities of unary formulas, expressed using arbitrary decidability logics over arithmetic.
In the process, we provide simpler proofs of some key prior  results on finite satisfiability and semi-linearity of the spectrum for these  logics.
\end{abstract}

%% file: 1-intro.tex
\section{Introduction}\label{sec:intro}
This paper concerns the  problem of determining whether a formula in some logic $\mathcal{L}$ is satisfied in some finite structure: the \emph{finite satisfiability problem for} $\mathcal{L}$. For the fragments of interest to us, the decidability of this problem will imply decidability of many other static analysis problems for the logic, such as validity of a sentence over finite structures and equivalence of two sentences over finite structures.
The problem is known
to be undecidable for first order logic~\cite{traktenbrot}. 
In the last decades two paradigms for decidability emerged, one based on restricting to \emph{guarded quantification}, and another restricting to \emph{two variable fragments}.
The two variable fragment of first order logic was shown to have decidable finite satisfiability problem and this was extended to allow ``counting quantifiers'',  of the form $\exists^{\ge k} y ~ \phi(x,y)$, where $k$ is fixed. The complexity of the satisfiability problem for this logic, normally denoted $\Ctwo$, over finite structures, as well as the complexity of satisfiability when considering arbitrary structures, was isolated in~\cite{ctwocomplexity}. The same problem for the sublogic where quantification is guarded was shown 
to have slightly lower complexity~\cite{gctwocomplexity}.

In recent years there has been interest in adding more powerful arithmetic capabilities to these decidable logics. Instead of just asking whether the number of elements satisfying a property is above or below a constant, 
we can ask whether the cardinality is even, or whether 
the cardinality of several properties satisfies a linear inequality. 
\emph{Presburger logic} focuses on linear inequalities between properties, 
and \emph{local Presburger logic}, $\GPtwo$, restricts to properties that involve neighbors of a given element. 
These logics are most naturally applied to finite models, to avoid dealing with addition or parity of infinity: 
thus in this work, \emph{finite satisfiability will always be our default}, and 
we just refer to ``satisfiability'' henceforward. 
In~\cite{usicalp} it was shown that local Presburger logics are in some sense equally expressive as 
certain variations of Graph Neural Networks (GNNs), and 
GNN verification problems can be reduced to satisfiability problems for these logics.
It has been shown that satisfiability of Presburger logic is undecidable, while satisfiability of local Presburger logic is decidable~\cite{percentage}: but the exact complexity of local Presburger logic had not been determined. Another way of enhancing logic with arithmetic involves \emph{global cardinality constraints}. 
In such an enhancement, we can write a constraint that defines a constraint (e.g., a linear inequality) 
between one variable formulas. But we cannot have such a constraint parameterized by a second free variable.  For example, we can write that unary predicate $U$ has cardinality below unary predicate $V$. But if $E$ and $F$ are two binary predicates, we cannot express that the elements that are not $E$ adjacent to $x$ have greater cardinality than the elements that are not $F$ adjacent to $x$.

As with local Presburger logic, it follows from prior results  that $\Ctwo$ enhanced with global cardinality constraints has decidable (finite) satisfiability problem.
The exact complexity is not explicitly studied in the literature,  although in the case of global cardinality constraints consisting of linear inequalities, tight bounds follow from results of Rudolph in~\cite{rudolphcardinality}.    A natural approach to proving decidability for logics
enhanced with global constraints follows is to show that the \emph{the spectrum of each tuple of unary formulas of the logic is effectively semilinear}. The spectrum of a tuple of unary formulas $\phi_1(x) \ldots \phi_k(x)$ is the set of possible values that the tuple of cardinalities of the satisfiers of the formula can have within a finite model. For example, for the pair of formulas $\phi_1(x)=U(x), \phi_2(x)= U(x) \wedge V(x)$ the spectrum consists of all pairs $(m, n)$ with $n \leq m$. The spectrum being effectively semilinear means that we can find a quantifier-free formula of arithmetic that defines it. From effective semilinearity of the spectrum, we easily obtain decidability with additional global constraints.
Here we provide  complexity bounds for several of these problems, one featuring  arithmetic over counts of neighbors, another concerning computing the spectrum of formulas. 
See Table~\ref{tab:results} for the summary of our new upper bounds.
At the same time, we provide simpler proofs of two fundamental results in the area: decidability of (finite) satisfiability for $\GPtwo$ (see Section~\ref{sec:gptwo}), and effective semi-linearity of the spectrum for $\Ctwo$ (Section~\ref{sec:ctwo}). We complement these positive results with a counterexample to semilinearity of spectra for $\GPtwo$ unary formulas (Section~\ref{sec:gptwocounter}).

%% file: 2-prelims.tex
\section{Preliminaries} \label{sec:prelims}

\myparagraph{Basic notions} A \emph{graph} will always mean a finite directed graph without self-loops.
A \emph{colored graph} is a graph where nodes have a unique node label, and every pair of distinct nodes is connected by a unique edge (in one direction or the other), with node and edge vocabularies distinct.
A \emph{multigraph} is a finite directed graph without self-loops that allows multiple edges (including no edge) between a pair of distinct nodes.
A \emph{colored multigraph} is defined similarly, but requiring each node to have a unique label.
The multiplicity of a multigraph is defined as the maximum edge multiplicity within it.

The $i^{th}$ unit vector, denoted by $\unit{i}$,
has $i^{th}$ entry $1$, and other entries are $0$.
We define $\bfzero$ as the vector where all entries are $0$,
and $\bfone$ as the vector that all entries are $1$.
For  $\bfv, \bfu \in \bbN^n$,
$\bfv \vle \bfu$ if, for $1 \le i \le n$, $\bfv_i \le \bfu_i$
Let $\projp{n}{\bfv}$ denote the projection operator that maps $\bfv$ to its first $n$ entries.
Let $\norm{\bfA}$ and $\norm{\bfv}$ denote the maximal absolute values of the entries in $\bfA$ and $\bfv$, respectively.

\myparagraph{Linear algebra and Kleene star}
An \emph{integer linear constraint} is an inequality $\sum_{i} a_i \cdot x_i \leq c$
involving integer coefficients $a_i$, non-negative integer variables $x_i$, and an integer constant $c$.
The set of solutions of such a constraint is called a \emph{modulus-free semilinear set}, while
if we include equations $x = c \mod p$ for integers $c, p>0$ it is a \emph{semilinear set}.
A constraint is \emph{homogeneous} if $c = 0$.

For an integer linear system $\cQ(\bfx): \bfA \bfx = \bfc$,
we denote by $\sol{\cQ(\bfx)}$ the set of solutions in $\bbN$,
i.e., $\sol{\cQ(\bfx)} := \setc{\bfv \in \bbN^n}{\bfA\bfv = \bfc}$, and similarly for a Boolean combination of linear systems.
Note that if $\cQ(\bfx)$ is feasible,
then it has a solution in which the number of non-zero assignments is bounded, and the values of the assignments are also bounded.

\begin{lemma}[\cite{papa}, \cite{caratheodory-integer}, also see Corollary 2.2 in~\cite{twovarpres}]\label{lem:ilp}
    There are constants $c_1, c_2 \in \bbN$ such that for every integer linear system $\cQ(\bfx)$,
    if $\cQ(\bfx)$ admits a solution in $\bbN$, then it admits a solution in $\bbN$ in which
    the number of variables assigned non-zero values is at most $c_1t\log\left(c_2tM\right)$ and
    every variable is assigned a value bounded by $c_1t(tM)^{c_2t}$,
    where $t$ is the number of constraints in $\cQ(\bfx)$ and $M$ is the maximal absolute value of the coefficients in $\cQ(\bfx)$.
\end{lemma}

%% file: 2a-kleenestar.tex
\begin{definition}[Kleene star]
    For a set $\cS \subseteq \bbN^n$,
    the \emph{Kleene star} of $\cS$ is defined as
    \begin{equation*}
        \kstar{\cS}\ :=\ \setc{\sum_{s \in \cS'} s}{\text{$\cS'$ is a finite multisubset of $\cS$}}.
    \end{equation*}
    Note that $\bfzero \in \kstar{\cS}$ for any $\cS$.
\end{definition}

If $\cS$ is semilinear then its Kleene star is semilinear: this is closely-related to Parikh's theorem which gives a correspondence between semi-linear sets and regular languages: see \cite{parikh,kleenestar}.
We will be interested in bounding the complexity of the Kleene star of a semilinear set using numerical data associated to the set.
Note that for an integer linear system $\cQ(\bfx): \bfA \bfx = \bfc$,
if $\bfc = \bfzero$, then $\kstar{\sol{\cQ(\bfx)}} = \sol{\cQ(\bfx)}$.
The following result gives a bound when $\bfc \neq \bfzero$.

\begin{lemma}\label{lem:key}
    For every integer linear system $\cQ(\bfx): \bfA \bfx = \bfc$ with $\bfc \neq \bfzero$,
    where $\bfA \in \bbZ^{m \times n}$ and $\bfc \in \bbZ^{m}$,
    there exists $\bfAtilde \in \bbZ^{(n+t) \times (n+k)}$ with $\norm{\bfAtilde} = 1$
    such that
    $\kstar{\sol{\cQ(\bfx)}} = \sol{\cQtilde(\bfx)}$,
    where $\cQtilde(\bfx)$ is the  Boolean combination of integer linear systems
    \begin{equation*}    
        \exists \bfy_1, \bfy_2\ 
        (\bfx\ =\ \projp{n}{\bfy_1} + \bfy_2)\ \land\ 
        (\bfAtilde\bfy_1\ =\ \bfzero)\ \land\ 
        (\bfA\bfy_2\ =\ \bfzero)\ \land\ 
        ((\bfy_1=\bfzero) \to (\bfy_2=\bfzero)),
    \end{equation*}
    $k = n\cdot(2D+1)^{m} + 1$,
    $t \le 2nk$, and
    $D = n \cdot \norm{\bfA} \cdot \left((n+1)\cdot \norm{\bfA} + \norm{\bfc} + 1\right)^{m}$.
    Moreover, $\bfAtilde$ can be computed in time $2^\bigO{\log n + \log K + m^2}$,
    where $K := \max\left(\norm{\bfA}, \norm{\bfc}\right)$.
\end{lemma}

We defer the proof of this to the appendix.

%% file: 2b-prelims-type.tex
\myparagraph{Types}
Let $\vocab$ be a vocabulary with unary  $U_1, \ldots, U_n$ and binary  $R_1, \ldots, R_m$.
\begin{definition}[1- and 2-type]
    A \emph{1-type} (w.r.t. $\vocab$) is a maximal consistent subset of
    \begin{equation*}
        \bigcup_{i \in \intsinterval{n}}\set{U_i(x), \neg U_i(x)}
        \ \cup\ 
        \bigcup_{i \in \intsinterval{m}}\set{R_i(x, x), \neg R_i(x, x)}.
    \end{equation*}
    A \emph{2-type} (w.r.t. $\vocab$) is a maximal consistent subset of
    \begin{equation*}
        \bigcup_{i \in \intsinterval{m}}
        \set{R_i(x, y), \neg R_i(x, y), R_i(y, x), \neg R_i(y, x)}.
    \end{equation*}
\end{definition}
We denote the set of $1$-types  by $\onetypes$ and the set of $2$-types $\twotypes$.

Given an element $v$ in a $\vocab$-structure $\cG$, we let $\onetypeof(v)$ denote its one-type in $\cG$
and similarly for a pair of elements $v, u$ in $\cG$ we let $\twotypeof(v,u)$ denote its two-type in $\cG$;
in both cases we omit the dependence on $\cG$ for brevity.

The following definitions are variations of notions
in Pratt-Hartmann's~\cite{ctwocomplexity,gctwocomplexity,ianrevisited}.

\begin{definition}[Silent 2-type]
    The \emph{silent type}, denoted by $\etanull$, is the unique 2-type that consists of only negated terms.
    If a type is not silent we say it is \emph{audible}.
\end{definition}

\begin{definition}[Dual of 2-type]
    For a 2-type $\eta$, the \emph{dual} of $\eta$, denoted by $\dual{\eta}$,
    it the unique 2-type that swaps $x$ and $y$ in $\eta$.
\end{definition}
We denote
the set of all 2-types except $\etanull$ by $\twotypes^+$, and
the set of all 2-types that contain $R_t(x, y)$ by $\twotypes_t$ for $1 \le t \le m$.

Note that $\abs{\Onetypes} = 2^{n+m}$,
$\abs{\twotypes} = 2^{2m}$, $\abs{\twotypes^+} = 2^{2m}-1$,
and $\abs{\twotypes_t} =2^{2m-1}$.

We can abstract a model based on types:
\begin{definition}[Type graph]
    We associate to a $\vocab$-structure in a binary vocabulary a colored graph in an exponentially larger vocabulary:
    each vertex is colored with a 1-type, and
    each edge is colored with a 2-type. This is the \emph{type graph} of the $\vocab$-structure.
\end{definition}

Losing information, we can abstract a finite model based on cardinalities of $1$-types. The notion will be crucial in Section \ref{sec:ctwo}:
\begin{definition}[$1$-type cardinality vector of a $\vocab$-structure] \label{def:cardboundvector}
    Let $\cG$ be a $\vocab$-structure.
    The $1$-type cardinality vector of $\cG$ is a $\abs{\onetypes}$-dimensional vector 
    where the $i^{th}$ component is the number of elements in $\cG$ with 1-type $\pi_i$.
\end{definition}

%% file: 3-logics.tex
\section{Logics and prior complexity bounds concerning them}

We fix a vocabulary $\vocab$ which consists of $n$ unary predicates $U_1, \ldots, U_n$ and $m$ binary predicates $R_1, \ldots, R_m$. All of our logics will consist of formulas with at most two variables, $x$ and $y$.
We start by reviewing \emph{two variable logic with counting} ($\Ctwo$), built up from atomic formulas $U_i(x)$, $ U_i(y)$, $R_i(x, y)$, and $R_i(y, x)$
via the Boolean operators and the quantifier constructs
$\rho(x) := \exists^{\geq k} y\ \phi(x,y)$,
where $\phi$ is a $\Ctwo$ formula and $k$ is a natural number, and similarly with the role of $x$ and $y$ swapped.

Our semantics will always be over \emph{finite structures} interpreting the vocabulary $\vocab$.
The formula $\rho(y)$ above holds in a model $M$ at an element $u$ for $x$ when there are least $k$ elements $v$ such that $M \models \phi(u, v)$.
The case $k = 1$ represents standard existential quantification.

\emph{Guarded two variable logic with counting} ($\GCtwo$), restricts $\Ctwo$ by requiring quantification to be guarded by an atom. The quantifier rule now includes the forms:
\begin{equation*}
    \exists^{\geq k} y\ R_i(x, y) \land \phi(x,y),
    \quad
    \exists^{\geq k} y\ R_i(y,x) \land \phi(x,y),
    \quad \text{or} \quad
    \exists^{\geq k} x\ ~ U_i(x) \land \phi(x).
\end{equation*}

The logic $\Ptwo$ replaces $\Ctwo$'s counting quantifier with a \emph{Presburger quantifier}:
\begin{equation*}
    \rho(x)\ :=\ \left(
    \sum_{t \in \intsinterval{m}} \kappa_{t} \cdot \presby{ \phi_t(x,y)}  \right)
    \ \circledast\ \delta,
\end{equation*}
where $\kappa_t$ are integers,
$\circledast$ is an inequality or equality, and $\delta$ is an integer.
In $\Ptwo$ we will also allow \emph{unguarded quantification over unary formulas}: if $\phi(x)$ is a formula, so is
$\exists x\ \phi(x)$.
For example,
$\exists x \left(\presby{R(x,y)} - 3 \cdot \presby{R(y,x)} = 0\right)$
is a sentence of $\Ptwo$ that holds in a directed graph  when there is an element that has 3 times as many outgoing edges as incoming edges.

The logic $\GPtwo$ (\emph{guarded two variable logic with Presburger quantifiers}) restricts $\Ptwo$ by requiring each $\phi_i$ in a Presburger quantifier to be guarded, although still allowing unguarded unary quantification.
Note that global cardinality constraints can be expressed in $\Ptwo$, but not in $\GPtwo$.
Following prior papers \cite{twovarpres}, we do not allow a quantifier that tests the modulus of a linear combination of
cardinalities of one dimensional sets, and thus our quantifiers are in a sense weaker than Presburger arithmetic.
A \emph{Presburger modulus constraint} is as above, but allowing $\circledast\ \delta$ to be replaced by 
$(= m \mod n)$, for natural numbers $m, n$ with $m < n$.
Likewise, $\Ctwoplusglobal$ can be extended to allow modulus constraints. 
In the appendix we show that \emph{all of of our results on $\GPtwo$ and $\Ctwoplusglobal$ also hold for the extension with modulus constraints}; in the body of the paper we focus on plain $\GPtwo$ and $\Ctwoplusglobal$ in the statements.

Given a unary formula $\phi(x)$ in some logic, and a finite model $M$, we let $|\phi|_M$ be the cardinality of $\phi$ in $M$; we extend this notation to a tuple $\langle \phi_1(x) \ldots \phi_k(x) \rangle$.
Given such a tuple, its \emph{spectrum} is the set of vectors of numbers of the form  $| \phi_1(x) \ldots \phi_k(x) |_M$. We say that such a spectrum is \emph{semi-linear} if it can be represented as a set of linear inequalities. When we talk about \emph{computing the spectrum} for a formula, we mean computing a semi-linear representation. 
If we have such a representation, we can clearly check whether an individual $\phi_i(x)$ is satisfiable, simply by seeing if the corresponding set of constraints is feasible. Thus, in cases, where the spectrum is known always to have such a representation, computing the spectrum can be seen as a generalization of the satisfiability pboelm.

The positive results in this paper, as well as the prior bounds, are summarized in Table \ref{tab:results}.
In the table, Sat Logic refers to the satisfiability problem for formulas of the logic over finite structures, while
Spectrum refers to a bound on the size of a representation of the spectrum, again over finite structure. An entry of the form $X$-c indicates that the problem is  complete for complexity class $X$.
\begin{tiny}
\begin{table}[!ht] \
\centering
\resizebox{\columnwidth}{!}{%
\begin{tabular}{l|l|l}
\hline
& Sat Logic     &   Spectrum Size Bound \\
\hline
$\GCtwo$ & $\exptime$-c \cite{gctwocomplexity} & {\bf $\expp$} (formerly
  $\twoexpp$) \\ 
$\Ctwo$ & $\nexptime$-c \cite{ctwocomplexity}  & {\bf $\expp$, Thm \ref{thm:spectrasemilinear}} (formerly  
$\twoexpp$ \cite{ultimatelyperiodic}) \\
$\GPtwo$ & {\bf $\exptime$-c, Thm \ref{thm:gptwoexp}} (formerly $\exptime$-hard, in $
\threenexp$ \cite{percentage}) & Open \\
$\Ptwo$ &  Undecidable \cite{percentage} & No size bound \\
\hline
\end{tabular}
}
\caption{Complexity results, with those new to this work in bold.}
\label{tab:results}
\end{table}
\end{tiny}

We give results on the spectrum, but they have applications to decidability of  richer logics. $\Ctwo$ can be extended with \emph{global unary cardinality constraints}: these are of
the form
$\rho := 
\left(\sum_{t \in \intsinterval{n}}
\kappa_t \cdot |\phi_t(x)|\right)\ \circledast\ \delta$,
where $\phi_t$ are $\Ctwo$ formulas in one free variable,
$\kappa_t$ are integers,
$\circledast$ is an inequality or equality,
and $\delta$ is an integer.
$\rho$ is a sentence, and we define the semantics only for finite structures $M$.
Such $M$ satisfies $\rho$ if,
letting $s_t$ be the number of elements in
$M$ satisfying $\phi_t(x)$, we have 
$\left(\sum_{t \in \intsinterval{n}} \kappa_t \cdot s_t\right)\ \circledast\ \delta$.
We write $\Ctwoplusglobal$ for the extension of $\Ctwo$ with global cardinality constraints.
As we explain later, computation of the spectrum provides an approach to get tight bounds on these extended logics. In the case where the global constraints are in Presburger arithmetic, these can be seen as an alternative to the technique of \cite{rudolphcardinality}, which obtains tight bounds for the logic through reduction to satisfiability of $\Ctwo$ itself.

%% file: 4-gptwo.tex
\section{Finite satisfiability of $\GPtwo$ is in $\expp$}\label{sec:gptwo}

The goal of this section is to prove:
\begin{theorem} \label{thm:gptwoexp}
    The finite satisfiability problem of $\GPtwo$ is in $\mexp$.
\end{theorem}

\myparagraph{$\GPtwo$ infrastructure}
We fix a vocabulary $\vocab$ which consists of $n$ unary predicates $U_1, \ldots, U_n$ and $m$ binary predicates $R_1, \ldots, R_m$.
We consider $\GPtwo$ sentences $\varphi$ over $\vocab$ in \emph{normal form}:
\begin{equation*}
    \varphi\ :=\ 
    \forall x\ \gamma(x)\ \land\ 
    \bigwedge_{i \in \intsinterval{m}} \forall x \forall y\ \left(R_i(x,y) \land x \neq y\right) \to \alpha_i(x, y)\ \land\ 
    \bigwedge_{i \in \intsinterval{n}} \forall x\ U_i(x) \to P_i(x),
\end{equation*}
where $\gamma(x)$ is a quantifier-free formula,
each $\alpha_i(x, y)$ is a quantifier-free formula, and
each $P_i(x)$ is a Presburger quantified formula of the form: 
\begin{equation*}
    P_i(x)\ :=\ \left(
        \sum_{t \in \intsinterval{m}} \lambda_{i, t} \cdot \presby{R_t(x, y) \land x \neq y}\right)
    \ \circledast_i\ \delta_i.
\end{equation*}
Note that in the normal form we allow
$
\forall x \forall y\ \left(R(x,y) \land x \neq y\right) \to \alpha(x, y)
$
which is equivalent to a $\GPtwo$  sentence.

There is a linear time algorithm that converts every $\GPtwo$ sentence into an equi-satisfiable sentence in normal form~\cite{twovarpres}, which is 
based on the standard renaming technique for $\GCtwo$ introduced in~\cite{Kazakov04,gctwocomplexity}.
See the appendix for an argument.

\begin{definition}[Compatibility of 1- and 2-types]
    For a $\GPtwo$ sentence $\varphi$ in the normal form above and a 1-type $\pi$,
    we say $\pi$ is \emph{compatible} with $\varphi$,
    if $\pi(x) \models \gamma(x)$, where $\gamma(x)$ is from the normal form,
    that is, the ``unary universal part'' of $\varphi$.

    For  1-types $\pi_1$ and $\pi_2$, and a 2-type $\eta$,
    we say that $\tuple{\pi_1, \eta, \pi_2}$ is \emph{compatible} with $\varphi$,
    if
    \begin{compactitem}
        \item $\pi_1(x) \land \eta(x, y) \land \pi_2(y) \models
        \bigwedge_{i \in \intsinterval{m}} \left(R_i(x,y) \land x \neq y\right) \to \alpha_i(x, y)$, and
        \item $\pi_2(x) \land \dual{\eta}(x, y) \land \pi_1(y) \models
        \bigwedge_{i \in \intsinterval{m}} \left(R_i(x,y) \land x \neq y\right) \to \alpha_i(x, y)$.
    \end{compactitem}
    That is, the tuple satisfies the ``binary universal part''.
\end{definition}
We denote the set of 1-types compatible with $\varphi$ by $\Onetypes^\varphi \subseteq \Onetypes$.
Note that it takes time $\bigO{\abs{\varphi}}$ to check compatibility with $\varphi$ either for a $1$-type $\pi$  or a tuple $\tuple{\pi_1, \eta, \pi_2} \in \Onetypes \times \twotypes \times \Onetypes$.

In order to capture the counting conditions, we introduce \emph{behavior vectors}.
Informally, the behavior vector of a vertex encodes the configuration of its neighbors.
We fix an ordering of the set $\twotypes^+ \times \Onetypes$,
which has cardinality $2^{n+3m}-2^{n+m}$.
A behavior vector $\bff$ is an element of $\bbN^{2^{n+3m}-2^{n+m}}$.
Abusing notation, we write $\bff(\pi, \eta)$
to refer to the $i^{th}$ component of $\bff$,
where $i$ is the index of the tuple $\tuple{\pi, \eta}$ in the fixed ordering of $\twotypes^+ \times \Onetypes$.

\begin{definition}[Behavior vector of a vertex]
    For a $\vocab$-structure $\cG$ and a vertex $v$ in $\cG$,
    the behavior vector realized by $v$,
    denoted by $\bh{v}$, is an element of $\bbN^{2^{n+3m}-2^{n+m}}$ defined as follows:
    for every $\eta \in \twotypes^+$ and $\pi \in \Onetypes$,
    $\bh{v}(\eta, \pi)$ is the number of edges $e$ in $\cG$ adjacent to $v$
    such that the 2-type of $e$ is $\eta$ and the 1-type of $e$'s destination is $\pi$.
\end{definition}

\begin{definition}[Characteristic system]
\label{def:characteristic-system}
    For $\GPtwo$ sentence $\varphi$ in normal form and 1-type $\pi$,
    the characteristic system of $\varphi$ and $\pi$, denoted $\cC^\varphi_\pi(\bfx)$,
    is the integer linear system defined:
    \begin{compactitem}
        \item The variables of $\cC^\varphi_\pi(\bfx)$ are $x_{\eta, \pi'}$
        for every $\eta \in \twotypes^+$ and $\pi' \in \Onetypes$.
        \item (Compatibility)
        For every $\eta \in \twotypes^+$ and $\pi' \in \Onetypes$
        if $\pi'$ or $\tuple{\pi, \eta, \pi'}$ is not compatible with $\varphi$,
        then $x_{\eta, \pi'} = 0$
        is a constraint in $\cC^\varphi_\pi(\bfx)$.
        \item (Counting)
        For $1 \le i \le n$, if $U_i(x) \in \pi$
        then the following constraint is in $\cC^\varphi_\pi(\bfx)$:
        \begin{equation*}
            \left(
                \sum_{t \in \intsinterval{m}}
                \lambda_{i, t} \cdot
                \sum_{\eta \in \twotypes_t}
                \sum_{\pi' \in \Onetypes} 
                x_{\eta, \pi'}
            \right)
            \ \circledast_i\ \delta_i.
        \end{equation*}
    \end{compactitem}

\end{definition}
Note that there are $\abs{\twotypes^+ \times \Onetypes}$ variables
and at most $n$, the number of unary predicates, constraints in $\cC^\varphi_\pi(\bfx)$.
The absolute value of coefficients in $\cC^\varphi_\pi(\bfx)$
is bounded by $K$, where $K$ is the maximal absolute value of coefficients in $\varphi$.

\begin{definition}[Compatibility of behavior vector]
    For a $\GPtwo$ sentence $\varphi$ in normal form 
    and a 1-type $\pi$,
    we say that a behavior vector $\bff$ is \emph{compatible} with $\varphi$ and $\pi$
    if $\bff$ is a solution of $\cC^\varphi_\pi(\bfx)$.
\end{definition}
We denote the set of all behavior vectors that are compatible with $\varphi$ and $\pi$
by $\cB^{\varphi}_{\pi} := \sol{\cC^\varphi_\pi(\bfx)}$.
Note that $\cB^{\varphi}_{\pi}$ can be infinite but semilinear.

The following property restates satisfaction of a formula in a $\vocab$-structure in terms of the type graph of the structure, where the types are identified with vertex- and edge-labels:
\begin{proposition}\label{prop:labeled_gptwo}
    For every $\GPtwo$ sentence $\varphi$ in normal form,
    $\varphi$ is finitely satisfiable if and only if
    there exists a colored graph $\cG$ satisfying the following conditions:
    \begin{compactitem}
        \item (1-type) For every vertex $v \in V$, $\onetypeof(v)$ is compatible with $\varphi$.
        \item (2-type) For every pair of vertices $v_1, v_2 \in V$,
        $\tuple{\onetypeof(v_1), \twotypeof(v_1, v_2), \onetypeof(v_2)}$ is compatible with $\varphi$.
        \item (Counting) For every vertex $v \in V$, $\bh{v}$ is compatible with $\onetypeof(v)$ and $\varphi$.
    \end{compactitem}
\end{proposition}

We say that a colored multiple graph in the vocabulary of types, $\cG$, is a \emph{finite pseudo-model} of a $\GPtwo$ sentence $\varphi$ if it satisfies all conditions in Proposition~\ref{prop:labeled_gptwo}.
Here we show that to establish the finite satisfiability of $\varphi$, it is sufficient to construct a pseudo-model.

\begin{lemma}\label{lem:pseudo}
    For every $\GPtwo$ sentence $\varphi$ in normal form,
    $\varphi$ is finitely satisfiable if and only if it has a finite pseudo-model.
\end{lemma}

\begin{proof}
    Note that the type graph of a finite model of $\varphi$ is a finite pseudo-model of $\varphi$. 
    Thus we focus on the if direction,  constructing (the type graph of) a finite model of $\varphi$ using a \emph{duplicating and swapping} procedure.
   
    Let $\cG_0$ be a finite pseudo-model of $\varphi$, and let $\cG'_0$ be a copy of $\cG_0$.
    Define $\cG_1$ as the disjoint union of $\cG_0$ and $\cG'_0$,
    after iteratively applying the following swapping procedure:
    for every pair of vertices $v$ and $u$ in $\cG_0$ with more than one edge between them,
    let $e$ be an edge between them, and let $e'$ be the corresponding edge in $\cG'_0$.
    We then swap the destinations of $e$ and $e'$ in $\cG_1$.

    It is straightforward to verify that $\cG_1$ is a pseudo-model of $\varphi$ and its multiplicity is one less than $\cG_0$.
    We can repeat this procedure until we obtain a finite pseudo-model $\cG_t$ where no two vertices have multiple edges.
    Finally, the type graph of the finite model  of $\varphi$ is obtained by connecting all remaining pairs in $\cG_t$ with the silent type.
\end{proof}

\myparagraph{The main construction}
For every $\GPtwo$ sentence $\varphi$ in normal form,
we assume that
for every 1-type $\pi \in \Onetypes^\varphi$,
$\bfzero \not\in \cB^{\varphi}_{\pi}$.
Otherwise $\varphi$ is trivially satisfiable
by a model with only one vertex with 1-type $\pi$. 

\begin{lemma}\label{lem:equiv1}
    A $\GPtwo$ sentence $\varphi$ in normal form is finitely satisfiable
    if and only if there exists vectors $\bfg_\pi \in \kstar{\cB^\varphi_\pi}$ for each $\pi \in \Onetypes^\varphi$ satisfying the following conditions. 
    \begin{compactitem}
        \item (Matching) For every $\tuple{\pi_1, \eta, \pi_2} \in \Onetypes^\varphi \times \twotypes^+ \times \Onetypes^\varphi$,
        $\bfg_{\pi_1}(\eta, \pi_2) = \bfg_{\pi_2}(\dual{\eta}, \pi_1)$.

        \item (Non-triviality) There exists some $\pi \in \Onetypes^\varphi$ such that $\bfg_\pi \neq \bfzero$.
    \end{compactitem}
\end{lemma}

\begin{proof}
    \textbf{\uline{Suppose that $\varphi$ is finitely satisfiable with finite model $\cG$}}.
    For a 1-type $\pi$,
    let $V_\pi \subseteq V$ be the set of vertices $v$ in $\cG$ with 1-type $\pi$.
    Note that the $V_\pi$ partition $V$.
    Since $\cG \models \varphi$, for every $v$ in $V_\pi$,
    $\bh{v} \in \cB^\varphi_\pi$.
    We claim that $\bfg_\pi = \sum_{v \in V_\pi} \bh{v} \in \kstar{\cB^\varphi_\pi}$ satisfies the conditions.

    \begin{compactitem}
        \item (Matching) 
        For every $\tuple{\pi_1, \eta, \pi_2} \in \Onetypes^\varphi \times \twotypes^+ \times \Onetypes^\varphi$,
        note that
        \begin{equation*}
            \begin{aligned}
                \sum_{v \in V_{\pi_1}} \bh{v}(\eta, \pi_2)
                \ =\ &
                \sum_{v \in V_{\pi_1}}
                \abs{\setc{u \in V}{\text{$v \neq u$, $\twotypeof(v, u) = \eta$, and $\onetypeof(u) = \pi_2$}}} \\
                \ =\ &
                \abs{\setc{(v, u) \in V^2}{\text{$v \neq u$, $\onetypeof(v) = \pi_2$, $\twotypeof(v, u) = \eta$, and $\onetypeof(u) = \pi_2$}}}.
            \end{aligned}
        \end{equation*}
        That is, $\bfg_{\pi_1}(\eta, \pi_2)$ is the number of edges with types $\tuple{\pi_1, \eta, \pi_2}$ in $\cG$.
        Similarly, $\bfg_{\pi_2}(\dual{\eta}, \pi_1)$ is the number of edges with types $\tuple{\pi_2, \dual{\eta}, \pi_1}$ in $\cG$.
        Since the number of edges with types $\tuple{\pi_1, \eta, \pi_2}$ and $\tuple{\pi_2, \dual{\eta}, \pi_1}$ should be the same, the Matching constraint holds.
        \item (Non-triviality)
        Since $\cG$ is non-empty, by our assumption,
        it has a vertex $v$ with 1-type $\pi$ satisfying that $\bh{v} \neq \bfzero$.
        Therefore,
        $\bfg_{\pi} = \bh{v} + \sum_{u \in V_{\pi} \setminus \set{v}} \bh{u} \neq \bfzero$.
    \end{compactitem}

    \textbf{\uline{Suppose that there exists such $\bfg_\pi \in \kstar{\cB^\varphi_\pi}$ for $\pi \in \Onetypes^\varphi$}}.
    By the definition of Kleene star,
    for every $\pi \in \Onetypes^\varphi$,
    there exists $\bff_{\pi, 1}, \ldots, \bff_{\pi, \ell_\pi} \in \cB^\varphi_\pi$,
    not necessarily distinct,
    such that $\bfg_\pi = \sum_{i \in \intsinterval{\ell_\pi}} \bff_{\pi, i}$. 
    We construct a colored multigraph $\cG$, which will be the type graph of our structure, as follows.
    The vertices in $\cG$ are
    $\bigcup_{\pi \in \Onetypes^\varphi}
    \set{v_{\pi, 1}, \ldots, v_{\pi, \ell_\pi}}$.
    Each vertex $v_{\pi, i}$ is colored with a $1$-type $\pi$,
    and we let $V_{\pi}$ be the set of vertices colored with $\pi$.
    For every $\tuple{\pi_1, \eta, \pi_2} \in \Onetypes^\varphi \times \twotypes^+ \times \Onetypes^\varphi$,
    we add edges between $V_{\pi_1}$ and $V_{\pi_2}$ 
    such that
    for each $v_{\pi_1, i}$, the number of edges connected to it with 2-type $\eta$ is $\bff_{\pi_1, i}(\eta, \pi_2)$ and 
    for each $v_{\pi_2, i}$, the number of edges connected to it with 2-type $\eta$ is $\bff_{\pi_2, i}(\dual{\eta}, \pi_1)$.
    Note that by the Matching constraint,
    $\bfg_{\pi_1}(\eta, \pi_2) = \bfg_{\pi_2}(\dual{\eta}, \pi_1)$,
    so this construction is always possible.
    We claim that $\cG$ is a pseudo-model of $\varphi$.
    Note that each vertex in $\cG$ is colored with 1-type in $\Onetypes^\varphi$,
    each edge in $\cG$ is colored with 2-type in $\twotypes^+$,
    and it is routine to check that the behavior vector of $v_{\pi, i}$ is $\bff_{\pi, i}$.
    By the Non-triviality constraint, $\cG$ is non-empty.
    Thus by Lemma~\ref{lem:pseudo}, $\varphi$ is finitely satisfiable.
\end{proof}

Above we connected the finite satisfiability problem of $\GPtwo$ sentence $\varphi$ and the Kleene star of behavior vectors.
However, because the cardinality of the set $\kstar{\cB^\varphi_\pi}$ is \emph{infinite},
it is not clear how  to determine the solvability of the conditions in Lemma~\ref{lem:equiv1}.
In addition, this representation does not tell us anything about the shape of the models.

Recall that the set of behavior vectors $\cB^\varphi_\pi$ is the solution set of the characteristic system $\cC^\varphi_\pi(\bfx)$.
We can apply Lemma~\ref{lem:key} and obtain that
\begin{equation*}
    \kstar{\cB^\varphi_\pi}
    \ =\ 
    \kstar{\sol{\cC^\varphi_\pi(\bfx)}}
    \ =\ 
    \sol{\cCtilde^\varphi_\pi(\bfx)}.
\end{equation*}
Thus the conditions in Lemma~\ref{lem:equiv1} can be encoded in the following integer linear system.

\begin{definition}
\label{def:linear-system-Q}
    For every $\GPtwo$ sentence $\varphi$ in normal form,
    the integer linear system $\cQ^\varphi(\bfx)$ is defined as follows.
    \begin{compactitem}
        \item The variables of $\cQ^\varphi(\bfx)$ are $x_{\pi_1, \eta, \pi_2}$
        for every $\pi_1, \pi_2 \in \Onetypes^\varphi$ and $\eta \in \twotypes^+$.
        Let $\bfx_{\pi_1}$ denote the vector of variables $x_{\pi_1, \eta, \pi_2}$.
        
        \item (Structure)
        For every $\pi \in \Onetypes^\varphi$,
        $\cCtilde^\varphi_\pi(\bfx_\pi)$
        is part of $\cQ^\varphi(\bfx)$,
        where $\cCtilde^\varphi_\pi(\bfx)$ is the integer linear system obtained by applying Lemma~\ref{lem:key} to the characteristic system $\cC^\varphi_\pi(\bfx)$.
            
        \item (Matching)
        For every $\tuple{\pi_1, \eta, \pi_2} \in \Onetypes^\varphi \times \twotypes^+ \times \Onetypes^\varphi$,
        the following constraint is in $\cQ^\varphi(\bfx)$,
        $x_{\pi_1, \eta, \pi_2} = x_{\pi_2, \dual{\eta}, \pi_1}$,
        where we recall that $\dual{\eta}$ denotes the dual 2-type.

        \item (Non-triviality)
            $\sum_{\pi_1 \in \Onetypes^\varphi}
            \sum_{\eta \in \twotypes^+}
            \sum_{\pi_2 \in \Onetypes^\varphi} x_{\pi_1, \eta, \pi_2}
             > 0$.
    \end{compactitem}
\end{definition}

\begin{lemma}\label{lem:equiv2}
    For every $\GPtwo$ sentence $\varphi$ in normal form,
    $\varphi$ is finitely satisfiable if and only if
    $\cQ^\varphi(\bfx)$ has a solution in $\bbN$.
\end{lemma}

\begin{proof}
    \textbf{\uline{Suppose that $\varphi$ is finitely satisfiable}}.
    By Lemma~\ref{lem:equiv1},
    there exists $\bfg_\pi \in \kstar{\cB^\varphi_\pi}$ for $\pi \in \Onetypes^\varphi$
    satisfying the Matching and Non-triviality conditions.

    For every $\pi_1, \pi_2 \in \Onetypes^\varphi$ and $\eta \in \twotypes^+$,
    let
    $a_{\pi_1, \eta, \pi_2} := \bfg_{\pi_1}(\eta, \pi_2)$.
    We claim that assigning  $x_{\pi_1, \eta, \pi_2}$ to $a_{\pi_1, \eta, \pi_2}$ gives a solution of the system $\cQ^\varphi(\bfx)$.
    \begin{compactitem}
        \item (Structure)
        By definition, for every $\pi \in \Onetypes^\varphi$,
        $\bfg_\pi \in \kstar{\cB^\varphi_\pi}$.
        By Lemma~\ref{lem:key},
        $\kstar{\cB^\varphi_\pi}
        = \kstar{\sol{\cC^\varphi_\pi(\bfx)}}
        = \sol{\cCtilde^\varphi_\pi(\bfx)}$.
        Thus, $\bfa_\pi$ is a solution of $\cCtilde^\varphi_\pi(\bfx)$.
        
        \item (Matching) For every $\tuple{\pi_1, \eta, \pi_2} \in \Onetypes^\varphi \times \twotypes^+ \times \Onetypes^\varphi$,
        by the Matching condition in Lemma~\ref{lem:equiv1},
        $\bfg_{\pi_1}(\eta, \pi_2) = \bfg_{\pi_2}(\dual{\eta}, \pi_1)$.
        Thus $a_{\pi_1, \eta, \pi_2} = a_{\pi_2, \dual{\eta}, \pi_1}$.

        \item (Non-triviality) By the Non-triviality condition in Lemma~\ref{lem:equiv1},
        there exists a 1-type $\pi \in \Onetypes^\varphi$
        such that $\bfg_\pi \neq \bfzero$.
        Therefore, 
        \begin{equation*}
            \begin{aligned}
                \sum_{\pi_1 \in \Onetypes^\varphi}
                \sum_{\eta \in \twotypes^+}\sum_{\pi_2 \in \Onetypes^\varphi} a_{\pi_1, \eta, \pi_2}
                \ =\ &
                \sum_{\pi_1 \in \Onetypes^\varphi}
                \sum_{\eta \in \twotypes^+}\sum_{\pi_2 \in \Onetypes^\varphi} \bfg_{\pi_1}(\eta, \pi_2) \\
                \ \ge\ &
                \sum_{\eta \in \twotypes^+}\sum_{\pi_2 \in \Onetypes^\varphi} \bfg_{\pi}(\eta, \pi_2) 
                \ >\ 0.
            \end{aligned}
        \end{equation*}
    \end{compactitem}

    \textbf{\uline{Suppose that $\cQ^\varphi(\bfx)$ has a solution, assigning  $x_{\pi_1, \eta, \pi_2}$ to $a_{\pi_1, \eta, \pi_2}$}}.
    For every $\pi \in \Onetypes^\varphi$,
    since $\bfa_{\pi}$ is a valid assignment of $\cCtilde^\varphi_\pi(\bfx)$,
    by Lemma~\ref{lem:key},
    $\bfa_{\pi} \in \kstar{\cB^\varphi_{\pi}}$.
    We claim that $\bfa_{\pi}$ are vectors satisfying the conditions in Lemma~\ref{lem:equiv1},
    which implies that $\varphi$ is finitely satisfiable.
    The proof of Non-triviality is routine. 
    Here we focus on the Matching constraint.
    For every $\tuple{\pi_1, \eta, \pi_2} \in \Onetypes^\varphi \times \twotypes^+ \times \Onetypes^\varphi$,
    by the Matching condition of $\cQ^\varphi$,
    $a_{\pi_1, \eta, \pi_2} = a_{\pi_2, \dual{\eta}, \pi_1}$.
    Since $\bfa_{\pi_1}(\eta, \pi_2) = a_{\pi_1, \eta, \pi_2}$,
    we have $\bfa_{\pi_1}(\eta, \pi_2) = \bfa_{\pi_2}(\dual{\eta}, \pi_1)$.
\end{proof}

The system $\cQ^\varphi(\bfx)$ is a Boolean combination of exponentially many integer linear constraints, and thus can be solved in
$\nexptime$. To get $\exptime$, we apply an idea from \cite{gctwocomplexity}, refining to get exponentially many homogeneous constraints.
\begin{definition}
    For every $\GPtwo$ sentence $\varphi$ in normal form,
    for $M \in \bbN$,
    the integer linear system $\cQ_M'^\varphi(\bfx)$ is the system $\cQ^\varphi(\bfx)$
    where the Structure constraint is replaced with:
    \begin{compactitem}
        \item (Structure$'$)
            For every $\pi \in \Onetypes^\varphi$,
            \begin{equation*}
                \exists \bfy_1, \bfy_2\ 
                (\bfx_\pi\ =\ \projp{n}{\bfy_1} + \bfy_2)\ \land\ 
                (\bfAtilde\bfy_1\ =\ \bfzero)\ \land\ 
                (\bfA\bfy_2\ =\ \bfzero)\ \land\ 
                ((\bfone \cdot \bfy_2) \le M (\bfone \cdot \bfy_2)),
            \end{equation*}
            where $\bfAtilde$ is the matrix obtained by applying Lemma~\ref{lem:key}
            to the characteristic system $\cC^\varphi_\pi(\bfx)$.
    \end{compactitem}
\end{definition}

\begin{lemma}\label{lem:equiv3}
    For every $\GPtwo$ sentence $\varphi$ in normal form,
    there exists $M^\varphi \in \bbN$, such that
    for every $M \ge M^\varphi$,
    $\cQ^\varphi(\bfx)$ has a solution in $\bbN$ if and only if
    $\cQ_{M}'^\varphi(\bfx)$ has a solution in $\bbN$.

    Moreover, $M^\varphi$ can be computed directly from $\varphi$ in time exponential in the length of $\varphi$, assuming binary encoding of coefficients.
\end{lemma}

\begin{proof}
    By Lemma~\ref{lem:ilp},
    if $\cQ^\varphi$ has a solution in $\bbN$,
    then it has a ``small solution'': every value is bounded by $M_0^\varphi = c_1t\left(t\maxlambda\right)^{c_2t}$,
    where $\maxlambda$ is the maximal coefficient in $\varphi$ and $t$ is the number of constraints in $\cQ^\varphi$, which is exponential in the number of unary and binary predicates.
    Thus $M_0^\varphi$ can be computed directly from $\varphi$ in time exponential in the length of $\varphi$, assuming binary encoding of coefficients.
    Let $M^\varphi = 2^{n}M_0^\varphi$.

    Note that the only difference between $\cQ^\varphi(\bfx)$ and $\cQ_{M}'^\varphi(\bfx)$
    pertains to constraints
    $(\bfy_1 = \bfzero) \to (\bfy_2 = \bfzero)$ and
    $ ((\bfone \cdot \bfy_2) \le M (\bfone \cdot \bfy_1))$.
    We show that the two systems are equi-satisfiable, provided $M$ is sufficiently large.

    \textbf{\uline{Suppose that $\cQ^\varphi(\bfx)$ has a solution in $\bbN$}}.
    We claim that a small solution $\bfa_{\pi}$ of $\cQ^\varphi(\bfx)$ 
    is a solution of $\cQ'^\varphi(\bfx)$.
    For every $\pi \in \Onetypes^\varphi$, 
    let $\bfb_1$ and $\bfb_2$ be corresponding assignments for $\bfy_1$ and $\bfy_2$ in $\cCtilde^\varphi_\pi(\bfx_\pi)$.
    If $\bfa_\pi = \bfzero$, then $\bfb_1 = \bfzero$ and $\bfb_2 = \bfzero$, 
    which implies that
    $(\bfone \cdot \bfb_2) \le M^\varphi (\bfone \cdot \bfb_1)$ holds trivially.
    On the other hand, if $\bfa_\pi \neq \bfzero$,
    then $\bfb_1 \neq \bfzero$, which implies that $(\bfone \cdot \bfb_1) \ge 1$.
    Therefore
    $(\bfone \cdot \bfb_2)
    \ \le\ 
    2^n M_0^\varphi
    \ =\ 
    M^\varphi
    \ \le\ 
    M
    \ \le\ 
    M (\bfone \cdot \bfb_1)$

    \textbf{\uline{Suppose that $\cQ_M'^\varphi(\bfx)$ has a solution $\bfa_{\pi}$ in $\bbN$}}.
    We claim that  the $\bfa_{\pi}$ also form a solution to $\cQ^\varphi(\bfx)$.
    It is sufficient to show that,
    for  vectors $\bfb_1$ and $\bfb_2$,
    if $(\bfone \cdot \bfb_2) \le M (\bfone \cdot \bfb_1)$,
    then $(\bfb_1 = \bfzero) \to (\bfb_2 = \bfzero)$.
    If $\bfb_1 = \bfzero$, then $(\bfone \cdot \bfb_2) \le M (\bfone \cdot \bfb_1) = 0$,
    which implies that $\bfb_2 = \bfzero$.
    Otherwise if $\bfb_1 \neq \bfzero$, 
    then $(\bfb_1 = \bfzero) \to (\bfb_2 = \bfzero)$ holds trivially.
\end{proof}

We are now ready to prove Theorem \ref{thm:gptwoexp}:

\begin{proof}
    For every $\GPtwo$ sentence $\varphi$ in normal form,
    by Lemma~\ref{lem:equiv1}, Lemma~\ref{lem:equiv2}, and Lemma~\ref{lem:equiv3},
    $\varphi$ is finitely satisfiable if and only if $\cQ_{M^\varphi}'^\varphi(\bfx)$ has a solution in $\bbN$.
    Note that $\cQ_{M^\varphi}'^\varphi(\bfx)$ is an integer linear system with 
    $2^\bigO{m\left(n + \log K\right)}$ variables, 
    $2^\bigO{m\left(n + \log K\right)}$ constraints, and
    coefficients bounded by $M^\varphi$.
    Moreover, $\cQ_{M^\varphi}'^\varphi(\bfx)$ is \emph{homogeneous}. 
    Thus, it has a solution in the natural numbers if and only if it has a \emph{rational} solution.
    The feasibility of an integer linear system over the rationals can be checked in time polynomial in the number of variables, the number of equations, and the logarithm of the absolute value of the maximum coefficients~\cite{IP}. Consequently, this process takes time exponential in the length of $\varphi$, assuming binary encoding of coefficients
\end{proof}

Note that the solutions to the linear system used in proving Theorem \ref{thm:gptwoexp} will contain a representation of every finite model:  the vector of cardinalities of $1$-types. But it contains additional vectors that do not correspond to the cardinalities of one types in any model. This is related to the fact that while existence of a pseudo-model implies existence of a model (Lemma \ref{lem:pseudo}), cardinalities are not preserved in moving from a pseudo-model to the corresponding model.  We return to this in Section \ref{sec:gptwocounter}.

%% file: 5-ctwo.tex
\section{Spectra of $\Ctwo$ formulas} \label{sec:ctwo}



In this section we will consider converting formulas in logic with arithmetic into linear constraints. These constraints will characterize  the set of models of a $\Ctwo$ formula. With this characterization in place, we have reduced finite satisfiability of the formula to a check for satisfiability of the conjunction of the linear constraints, which can be done via standard linear algebra solving. From this it will follow that $\Ctwo$ with a variety of global cardinality constraints is decidable: we just conjoined the cardinality constraints with the constraints characterizing the models.

We fix a vocabulary $\vocab$ which consists of $n$ unary predicates $U_1, \ldots, U_n$ and $m$ binary predicates $R_1, \ldots, R_m$.
We consider $\Ctwo$ sentence $\varphi$ over $\vocab$ in \emph{normal form}:
\begin{equation*}
    \varphi\ :=\ 
    \forall x\ \gamma(x)\ \land\ 
    \forall x\ \forall y\ \left(x \neq y \to \alpha(x, y)\right)\ \land\ 
    \bigwedge_{i \in \intsinterval{m'}} \forall x\ \exists^{=k_i} y\ \left(R_i(x,y) \land x \neq y\right),
\end{equation*}
where $\gamma(x)$ and $\alpha(x, y)$ are quantifier-free,
$R_i$ are binary predicates,
$0 \le m' \le m$,
and $k_i \in \bbN$.

\begin{definition}[Count bound vector]
The count bound vector of $\varphi$, denoted by $\bfk^\varphi$, 
is the $m'$ dimension vector such that
the $i^{th}$ component of $\bfk^\varphi$ is $k_i$.
\end{definition}

We denote the sum of $k_i$ by $K^\varphi$.
\begin{definition}[Forward and backward characteristic vector of 2-types]
    For $m'$ as above and a 2-type $\eta$, the forward characteristic vector of $\eta$,
    denoted by $\eta^\trif$, is an $m'$ dimensional vector defined:
    for every $1 \le i \le m'$, if $R_i(x, y) \in \eta$,
    then the $i^{th}$ component of $\eta^\trif$ is 1.
    Otherwise, it is $0$.
    The backward characteristic vector of $\eta$,
    denoted by $\eta^\trib$,
    is defined similarly.
\end{definition}

\begin{definition}[Forward and backward silent 2-type]
    For $m'$ as above and a 2-type $\eta$ 
    we say that $\eta$ is forward-silent if for $1 \le i \le m'$,
    $\neg R_i(x, y) \in \eta$.
    Backward-silent is defined similarly.
\end{definition}

We define the notions of $\Onetypes^\varphi$ and $\twotypes^\varphi_{\pi_1, \pi_2}$ similarly to those in $\GPtwo$.
We also define
$\cKvpnull_{\pi_1, \pi_2} \subseteq \twotypes^\varphi_{\pi_1, \pi_2}$ as the set of 2-types that are both forward- and backward-silent,
$\cKvpf_{\pi_1, \pi_2}$ for 2-types that are not forward-silent but are backward-silent,
and
$\cKvpfb_{\pi_1, \pi_2}$ for 2-types that are neither forward- nor backward-silent.

In the rest of the section, we deal only with normal form sentences. We explain how the results apply to arbitrary $\Ctwo$ sentences in the appendix.

\begin{definition}[Spectrum and $1$-type spectrum]
    For every $\Ctwo$ sentence $\varphi$,
    the spectrum of $\varphi$, denoted by $\spec(\varphi)$,
    is the set of natural numbers that are sizes of finite models of $\varphi$.

    The 1-type spectrum of $\varphi$, denoted by $\pispec(\varphi)$,
    is the set of 1-type cardinality vectors (see Definition~\ref{def:cardboundvector}) corresponding to finite models of $\varphi$.
\end{definition}

We aim to show: 

\begin{theorem} \label{thm:spectrasemilinear}
    For every $\Ctwo$ sentence $\varphi$,
    $\spec(\varphi)$ and $\pispec(\varphi)$ are semilinear.
    Furthermore, each spectrum can be represented by an existential Presburger formula of size doubly exponential in $\varphi$.
\end{theorem}

\begin{definition}[Silent compatible 1-types]
    For a pair of 1-types $\pi_1$ and $\pi_2$,
    we say that $\pi_1$ and $\pi_2$ are \emph{silent-compatible} w.r.t. $\varphi$,
    if $\cKvpnull_{\pi_1, \pi_2}$ is non-empty. Otherwise, we say, following~\cite{ctwocomplexity}, that 
    $\pi_1$ and $\pi_2$ are a \emph{noisy pair}.
\end{definition}

\begin{lemma} \label{lem:differentiated}
    For every $\Ctwo$ sentence $\varphi$ in normal form and $\vocab$-structure $\cG \models \varphi$,
    for every pair of 1-types $\pi_1$ and $\pi_2$,
    if $\pi_1$ and $\pi_2$ are a noisy pair (w.r.t. $\varphi$), then
    at least one of the 1-types $\pi_1$ or $\pi_2$ is realized by at most $2K^\varphi+1$ vertices in $\cG$.
\end{lemma}

\begin{proof}
    
    Recall that $V_{\pi}$ denotes the vertices in $V$ with 1-type $\pi$.
    Because $\cG \models \varphi$, for vertex $v \in \cG$,
    $\sum_{u \in V \setminus \set{v}} \twotypeof^\trif(v, u) = \bfk^\varphi$.
    We consider two cases.

    \textbf{\uline{Suppose that $\pi_1 = \pi_2$.}}
    Since $\twotypeof^\trif(v, u) = \twotypeof^\trib(u, v)$,
    we have
    \begin{equation*}
        \begin{aligned}
            \sum_{v \in V_{\pi_1}} \sum_{u \in V_{\pi_1} \setminus \set{v}}
            (\twotypeof^\trif(v, u) + \twotypeof^\trib(v, u))
            \ =\ &
            2 \sum_{v \in V_{\pi_1}} \sum_{u \in V_{\pi_1} \setminus \set{v}}
            \twotypeof^\trif(v, u) \\
            \ \vle\ &
            2 \sum_{v \in V_{\pi_1}} \sum_{u \in V \setminus \set{v}} \twotypeof^\trif(v, u)
            \ =\ 2\abs{V_{\pi_1}} \bfk^\varphi.
        \end{aligned}
    \end{equation*}
    For every pair of vertices $v$ and $u$ in $V_{\pi_1}$ with $v \neq u$,
    by definition of noisy pair,
    the sum of components of $\twotypeof^\trif(v, u) + \twotypeof^\trib(v, u)$ is at least $1$.
    Therefore, we have
    \begin{equation*}
        \sum_{v \in V_{\pi_1}} \sum_{u \in V_{\pi_1} \setminus \set{v}} 1
        \ =\ \abs{V_{\pi_1}}\cdot(\abs{V_{\pi_1}} - 1)
        \ \le\ 2\abs{V_{\pi_1}} \cdot K^\varphi,
    \end{equation*}
    which implies that $\abs{V_{\pi_1}} \le 2K^\varphi + 1$. 
    
    \textbf{\uline{Suppose that $\pi_1 \neq \pi_2$.}}
    By a similar argument, 
    $\abs{V_{\pi_1}} \cdot \abs{V_{\pi_2}} \le (\abs{V_{\pi_1}} + \abs{V_{\pi_2}}) \cdot K^\varphi$.
    If $\abs{V_{\pi_1}} \leq 2K^\varphi+1$ we are done. So assume the opposite.
    Then we have
    \begin{equation*}
        \abs{V_{\pi_2}}
        \ \le\ 
        \frac{K^\varphi}{1-\frac{K^\varphi}{\abs{V_{\pi_1}}}}
        \ \le\ 
        \frac{K^\varphi}{1-\frac{K^\varphi}{2K^\varphi + 1}}
        \ \le\ 
        2K^\varphi.
    \end{equation*}
\end{proof}

The following straightforward proposition gives
conditions that characterize the satisfiability of $\Ctwo$ by a model in terms of the prior constructs.
\begin{proposition}\label{prop:ctwo_graph}
    For every $\Ctwo$ sentence $\varphi$ in normal form and $\vocab$-structure $\cG$,
    $\cG$ is a model of $\varphi$ if and only if following conditions hold.
    For every $v, u \in V$ with $v \neq u$,
    \begin{compactitem}
        \item (1-type) $\onetypeof(v) \in \Onetypes^\varphi$,
        \item (2-type) $\twotypeof(v, u) \in \twotypes^\varphi_{\onetypeof(v), \onetypeof(u)}$, and
        \item (Counting) $\sum_{u \in V \setminus \set{v}} \twotypeof^\trif(v, u) = \bfk^\varphi$,
    \end{compactitem}
    where $\bfk^\varphi$ is the count bound vector of $\varphi$.
\end{proposition}

\begin{definition}[Partial model of $\Ctwo$ sentence] \label{def:partial}
    For every $\Ctwo$ sentence $\varphi$ and $\vocab$-structure $\cG$,
    we say that a complete colored graph $\cG$ is a \emph{partial model} of $\varphi$
    if $\cG$ satisfies the first two conditions in Proposition~\ref{prop:ctwo_graph}.
\end{definition}

The following is obvious.
\begin{lemma}\label{lem:core_partial_model}
    For every $\Ctwo$ sentence $\varphi$ in normal form,
    for every $\vocab$-structure $\cG \models \varphi$,
    for every $\cH \subseteq \cG$,
    $\cH$ is a partial model of $\varphi$.
\end{lemma}

\myparagraph{Core of structures}
\begin{definition} \label{def:core}
    For $\cG$ a $\vocab$-structure, $\ell \in \bbN$,
    we say that substructure $\cH \subseteq \cG$ is an $\ell$-core of $\cG$ if 
    \begin{compactitem}
        \item for every 1-type $\pi$,
        the number of vertices in $\cG \setminus \cH$
        with 1-type $\pi$ is either $0$ or at least $\ell$
        \item for  1-types $\pi_1$, $\pi_2$, and 2-type $\eta$,
        there are either $0$ and at least $\ell$
        pairs of vertices $v, u \in \cG \setminus \cH$ satisfying that
        $v \neq u$,
        $\onetypeof(v) = \pi_1$,
        $\twotypeof(v, u) = \eta$,
        and $\onetypeof(u) = \pi_2$.
    \end{compactitem}
\end{definition}

Note that the core of $\cG$ is not unique,
and $\cG$ is trivially a core of itself for every $\ell \in \bbN$.
Therefore, we are interested in the existence of a \emph{small} core of $\cG$.
We can show that for every $\ell$, every $\vocab$-structure has an \emph{exponential-size} core.

\begin{lemma}\label{lem:small_core}
    For every $\vocab$-structure $\cG$, for every $\ell \in \bbN$,
    $\cG$ has an $\ell$-core $\cH$ with size at most $(2^{2n+4m+2}) \cdot \ell$.
\end{lemma}

The intuition behind the construction is to repeatedly select a 1-type or tuple that violates the $\ell$-core condition
and add the relevant vertices to $\cH$.
Since there are only finitely many such choices, the process eventually stabilizes,
and the size of the constructed core remains bounded: see the appendix for details.
 
For a $\Ctwo$ sentence $\varphi$ and a $\vocab$-structure $\cG \models \varphi$,
if $\cH$ is a $\ell$-core of $\cG$ with sufficiently large $\ell$,
then we have the following important property of $\cG \setminus \cH$:

\begin{lemma}\label{lem:mutually_null_compatible}
    For every $\Ctwo$ sentence $\varphi$ and $\vocab$-structure $\cG \models \varphi$,
    for every $\ell$-core $\cH \subseteq \cG$ with $\ell > 2K^\varphi + 1$, 
    the 1-types realized in $\cG \setminus \cH$ are silent-compatible (w.r.t. $\varphi$).
\end{lemma}

\begin{proof}
    Suppose that there exist 1-types $\pi_1$ and $\pi_2$ realized in $\cG \setminus \cH$
    which are not  silent-compatible.
    By Lemma~\ref{lem:differentiated},
    at least one of the 1-types $\pi_1$ or $\pi_2$ is realized by at most $2K^\varphi+1$ vertices in $\cG$.
    Without loss of generality,
    suppose that $\pi_1$ is realized by at most $2K^\varphi+1$ vertices in $\cG$,
    then there are at most $2K^\varphi+1 < \ell$ vertices in $\cG \setminus \cH$ realized $\pi_1$.
    However, by the definition of $\ell$-core, since $\pi_1$ is realized in $\cG \setminus \cH$,
    it is realized by at least $\ell$ vertices in $\cG \setminus \cH$.
    This leads to a contradiction.
\end{proof}

In Section~\ref{sec:gptwo}, we introduced the notion of behavior vectors to describe the number of edges connected to a given vertex with specific 2-types.
In this section, we consider $\vocab$-structures with a fixed core and extend the concept of behavior vectors to incorporate information about the core.
Again, we fix an ordering of the set $\twotypes^+ \times \Onetypes$,
which has cardinality $2^{n+3m}-2^{n+m}$.

\begin{definition}[Extended behavior]
    Fixing a vocabulary $\vocab$, an \emph{extended behavior} w.r.t. 
    a set of vertices $V$ is a tuple $\tuple{g, \bff}$ where
    $g$ is a mapping from a vertex $v \in V$ to a 2-type in $\twotypes$, 
    and
    $\bff$ is an element of $\bbN^{2^{n+3m}-2^{n+m}}$.

\end{definition}

Informally, each extended behavior consists of two components.
The first one is a mapping from vertices in the core to 2-types,
capturing the 2-type realized by the edge between a given vertex and core vertices.
The second component is a behavior vector.

Abusing notation, for an extended behavior $\tuple{g, \bff}$,
we will write $g^\trif(v)$ for ${(g(v))}^\trif$ and
$\dual{g}(v)$ for the dual of $g(v)$.
We will also write
$\bff(\eta, \pi)$ for the $i^{th}$ component of $\bff$,
where $i$ is the index of the tuple $\tuple{\eta, \pi}$ in the fixed ordering of $\twotypes^+ \times \Onetypes$.

\begin{definition}[Compatibility of extended behaviors]\label{def:compatextended}
    For a $\Ctwo$ sentence $\varphi$ in normal form,
    $\pi \in \onetypes^\varphi$, and
    a $\vocab$-structure $\cH$ with vertices $V$,
    we say that an extended behavior function $\tuple{g, \bff}$ w.r.t. $V$ is
    \emph{compatible} with $\varphi$, $\pi$, and $\cH$
    if $f$ satisfies:
    \begin{enumerate}
        \item (Compatibility with $\cH$)
        For every $v \in V$, $g(v) \in \twotypes^\varphi_{\pi, \onetypeof(v)}$.

        \item (Compatibility with 1- and 2- types)
        For every $\tuple{\eta, \pi'} \in \twotypes^+ \times \Onetypes$,
        if $\pi' \not\in \onetypes^\varphi$ or 
        $\eta \not\in \twotypes^\varphi_{\pi, \pi'}$,
        then $\bff(\eta, \pi') = 0$.

        \item (Forward-silent zero)
        For every $\tuple{\eta, \pi'} \in \twotypes^+ \times \Onetypes$,
        if $\eta$ is forward-silent,
        then $\bff(\eta, \pi') = 0$.
       
        \item (Counting) The remaining entries of $\bff$,
        that is, entries for $2$-types that are not forward-silent, are bounded by the following constraint:
        \begin{equation*}
            \sum_{v \in V} g^\trif(v) + 
            \sum_{\pi' \in \Onetypes^\varphi}
            \left(
                \sum_{\eta \in \cKvpf_{\pi, \pi'}} \bff(\eta, \pi') \cdot \eta^\trif +
                \sum_{\eta \in \cKvpfb_{\pi, \pi'}} \bff(\eta, \pi') \cdot \eta^\trif
            \right)
            \ =\ \bfk^\varphi.
        \end{equation*}
        Note that this condition implies that
        if $\pi'$ and $\eta$ are compatible with $\varphi$, and
        $\eta$ is not forward-silent,
        then $\bff(\eta, \pi')$ is bounded by $K^\varphi$.
    \end{enumerate}
\end{definition}

Note that given an extended behavior,
we can check whether it is compatible with $\pi$ and $\cH$ in time polynomial in $2^n$ and $\abs{V}$.
Let $\cB^{\varphi}_{\pi, \cH}$ be the set of extended behaviors w.r.t. 
the vertices of $\cH$ that are compatible with $\varphi$, $\pi$ and $\cH$.
For each $\tuple{g, \bff} \in \cB^{\varphi}_{\pi, \cH}$,
recall that every entry of $\bff$ is bounded by $\abs{K^\varphi}$.
Thus the cardinality of $\cB^{\varphi}_{\pi, \cH}$ is bounded by
$\abs{V}^\abs{\cK} \cdot \abs{K^\varphi}^\abs{\twotypes \times \onetypes}$
which is doubly exponential in the length of $\varphi$.

\myparagraph{Computing a semilinear representation of the spectrum of a $\Ctwo$ sentence}
We are now ready to show that the $1$-type spectrum of a $\Ctwo$ sentence $\varphi$ in normal form is semilinear.
For every finite model $\cG$ of $\varphi$,
by Lemma~\ref{lem:differentiated},
for sufficiently large $\ell$,
for every $\ell$-core $\cH$ of $\cG$
and pair of 1-types $\pi_1$ and $\pi_2$ in $\cG \setminus \cH$,
$\pi_1$ and $\pi_2$ are silent-compatible w.r.t. $\varphi$. 
Furthermore by Lemma~\ref{lem:small_core},
$\cG$ has an $\ell$-core whose size is exponential in $\varphi$ and $\ell$.
We prove Theorem~\ref{thm:c2_spec} via a two-step construction.
First we show that for every such core $\cH$,
there exists an existential Presburger formula $\Psi^\varphi_\cH(\bfx)$
such that the solution set of $\Psi^\varphi_\cH(\bfx)$ corresponds to the set of 1-type cardinality vectors of finite models of $\varphi$ with $\ell$-core $\cH$.

\begin{compactitem}
    \item The variables are $x_\pi$ for $\pi \in \onetypes^\varphi$,
    and $y_{\pi, \tuple{g, \bff}}$ for $\pi \in \onetypes^\varphi$ and $\tuple{g, \bff} \in \cB^\varphi_{\pi, \cH}$.
    Intuitively, $\bfx$ is the 1-type spectrum of the finite model $\cG$ and $y_{\pi, \tuple{g, \bff}}$ is the number of vertices in $\cG \setminus \cH$ with 1-type $\pi$ and extended behavior $\tuple{g, \bff}$.
    \item
    The formula $\summ_{\cH}(\bfx, \bfy)$ asserts that the 1-type spectrum of $\cG$ is the sum of the 1-type spectrum of $\cH$ and
    the $1$-type spectrum of $\cG \setminus \cH$.
    Let $s_\pi$ be the number of vertices in $\cH$ with 1-type $\pi$. \\
    \begin{notsotiny}
    $\displaystyle\summ_{\cH}(\bfx, \bfy)
         := 
     \bigwedge_{\pi \in \onetypes^\varphi} \left(x_\pi = s_\pi + \sum_{\tuple{g, \bff} \in \cB^{\varphi}_{\pi, \cH}} y_{\pi, \tuple{g, \bff}}\right)$.
    \end{notsotiny}
    \item
    The formula $\comp_{\cH}(\bfy)$ guarantees that the vertices in $\cH$ satisfies the Counting condition of Proposition~\ref{prop:ctwo_graph}.
    Let $V_c$ be the vertices in $\cH$. \\
    \begin{notsotiny}
    $\displaystyle\comp_{\cH}(\bfy)
         := 
        \bigwedge_{v \in V_c} \left(
            \sum_{u \in V_c \setminus \set{v}} \twotypeof^\trif(v, u) +
            \sum_{\pi \in \onetypes^\varphi} \sum_{\tuple{g, \bff} \in \cB^{\varphi}_{\pi, \cH}} g^\trib(v) \cdot y_{\pi, \tuple{g, \bff}}
        = \bfk^\varphi \right)$.
        \end{notsotiny}
    \item
    The formula $\silent(\bfy)$ assert that 1-types realized in $\cG \setminus \cH$ are silent compatible. \\
    \begin{notsotiny}
    $\displaystyle\silent_{\cH}(\bfy)
         := 
        \bigwedge_{\substack{\pi_1, \pi_2 \in \onetypes^\varphi \\ \text{are not silent compatible}}}
        \left(\sum_{\tuple{g, \bff} \in \cB^{\varphi}_{\pi_1, \cH}} y_{\pi_1, \tuple{g, \bff}} = 0\right) \lor
        \left(\sum_{\tuple{g, \bff} \in \cB^{\varphi}_{\pi_2, \cH}} y_{\pi_2, \tuple{g, \bff}} = 0\right)$.
        \end{notsotiny}
    \item
    The formula $\bige_{\cH}(\bfy)$ guarantees that 1-types and tuples realized in $\cG \setminus \cH$ satisfy the requirements of a core. \\
    \begin{notsotiny}
    $
    \begin{aligned}
            \bige_{\cH}(\bfy)
        \ :=\ &
        \bigwedge_{\pi \in \onetypes^\varphi}
            \exists z\ 
            \left(
                z = \sum_{\tuple{g, \bff} \in \cB^\varphi_{\pi, \cH}} y_{\pi, \tuple{g, \bff}}
            \right) \land
            \left(
                z = 0  \lor
                z \ge {(2K^\varphi+1)}^2 
            \right)
        \land \\
        &\bigwedge_{\substack{
            \pi_1, \pi_2 \in \onetypes^\varphi \\ 
            \eta \in \cKvpfb_{\pi_1, \pi_2}}}
        \exists z\ 
        \left(z = \sum_{\tuple{g, \bff} \in \cB^\varphi_{\pi_1, \cH}} \bff(\eta, \pi_2) \cdot y_{\pi_1, \tuple{g, \bff}}\right) \land
        \left(
            z = 0 \lor
            z \ge (2K^\varphi+1)^2
            \right).
        \end{aligned}
    $
    \end{notsotiny}
    \item
        The formula $\match_{1, \cH}(\bfy)$ encodes the edge matching condition for audible 2-types and vertices in $\cG \setminus \cH$. \\
        \begin{notsotiny}
        $\displaystyle
        \match_{1, \cH}(\bfy)
         := 
        \bigwedge_{\substack{
            \pi_1, \pi_2 \in \onetypes^\varphi \\
            \eta \in \cKvpfb_{\pi_1, \pi_2}}}
        \left(
            \sum_{\tuple{g, \bff} \in \cB^{\varphi}_{\pi_1, \cH}} \bff(\eta, \pi_2) \cdot y_{\pi_1, \tuple{g, \bff}}
            =
            \sum_{\tuple{g, \bff} \in \cB^{\varphi}_{\pi_2, \cH}} \bff(\dual{\eta}, \pi_1) \cdot y_{\pi_2, \tuple{g, \bff}}
        \right)
        $.
        \end{notsotiny}

    \item
        The formula $\match_{2, \cH}(\bfy)$ encodes the edge matching condition for backward-silent but not forward-silent 2-types and vertices in $\cG \setminus \cH$. \\
        \begin{notsotiny}
        $\displaystyle
        \match_{2, \cH}(\bfy)
         := 
        \bigwedge_{\substack{
            \pi_1, \pi_2 \in \onetypes^\varphi \\
            \eta \in \cKvpf_{\pi_1, \pi_2}}}
        \left(
            \left(
                \sum_{\tuple{g, \bff} \in \cB^{\varphi}_{\pi_1, \cH}} \bff(\eta, \pi_2) \cdot y_{\pi_1, \tuple{g, \bff}} > 0
            \right)
            \to
            \left(
                \sum_{\tuple{g, \bff} \in \cB^{\varphi}_{\pi_2, \cH}} y_{\pi_2, \tuple{g, \bff}} > 0
            \right)
        \right)
        $.
        \end{notsotiny}
        Keep in mind that we will always be considering solutions in the natural numbers. So this implication could be rewritten without summation: if for one of the extended behaviors for this triple, $\bff(\eta, \pi_2) \cdot y_{\pi_1, \tuple{g, \bff}} > 0$, then  one of the numbers $y_{\pi_2, \tuple{g, \bff}} > 0$.
\end{compactitem}
Finally, $\displaystyle
    \Psi^\varphi_{\cH}(\bfx)
         :=
         \exists \bfy\ 
        \summ_{\cH}(\bfx, \bfy) \land
        \comp_{\cH}(\bfy) \land
        \silent_{\cH}(\bfy) \land
        \bige_{\cH}(\bfy) \land
        \bigwedge_{i=1, 2}
        \match_{i, \cH}(\bfy)$.

This system characterizes the spectrum:
\begin{lemma}\label{lem:c2_spec_core}
    For every $\Ctwo$ sentence $\varphi$ and partial model $\cH$ of $\varphi$,
    a vector $\bfv$ is a solution of $\Psi^\varphi_\cH(\bfx)$
    if and only if
    there exists a finite model $\cG$ of $\varphi$ such that
    the 1-type cardinality vector of $\cG$ is $\bfv$,
    and $\cH$ is a $(2K^\varphi+1)^2$-core of $\cG$.
\end{lemma}

It is easy to see that vectors in the spectrum satisfy the equations. 
In the other direction, we extend the core $\cH$ to a finite model $\cG$ of $\varphi$.
Solutions to the equations tell us how many vertices of each $1$-type and extended behavior to create in $\cG$.
We assign edge colors so that each vertex realizes its intended extended behavior.
This is always possible because the solution is sufficiently large, as ensured by the Big condition in the equation. This extra size gives us some leeway to match up $1$-types via $2$-types.
Details are in the appendix.

Let $\mathfrak{H}$ be the collection of partial models of $\varphi$ with size at most $2^{2n+4m+2} \cdot (2K^\varphi+1)^2$.
We define the existential Presburger formula 
$\Psi^\varphi(\bfx)
 := 
\bigvee_{\cH \in \mathfrak{H}}
\Psi^\varphi_{\cH}(\bfy).$

\begin{lemma}\label{thm:c2_spec}
    For every $\Ctwo$ sentence $\varphi$,
    $\sol{\Psi^\varphi(\bfx)} = \pispec(\varphi)$.
\end{lemma}

\begin{proof}
    Suppose that $\bfv \in \sol{\Psi^\varphi(\bfx)}$,
    there exists $\cH$ such that $\bfv \in \sol{\Psi^\varphi_\cH(\bfx)}$.
    By Lemma~\ref{lem:c2_spec_core},
    there exists a model $\cG$ of $\varphi$ such that the 1-type cardinality vector of $\cG$ is $\bfv$.
    Thus $\bfv \in \pispec(\varphi)$.
    On the other hand, suppose that $\bfv \in \pispec(\varphi)$.
    Then there is a model $\cG$ of $\varphi$ with 1-type cardinality vector  $\bfv$.
    By Lemma~\ref{lem:small_core}, $\cG$ has a $(2K^\varphi+1)^2$-core $\cH$ with size at most $2^{2n+4m+2}(2K^\varphi+1)^2$,
    which implies that $\cH \in \mathfrak{H}$.
    By Lemma~\ref{lem:c2_spec_core},
    $\bfv \in \sol{\Psi^\varphi_\cH(\bfx)}$.
    Since $\cH \in \mathfrak{H}$,
    $\bfv \in \sol{\Psi^\varphi(\bfx)}$.
\end{proof}

Theorem~\ref{thm:spectrasemilinear} follows immediately from Lemma~\ref{thm:c2_spec}.

Note that the length of the formula $\Psi^\varphi(\bfx)$ is doubly exponential in the length of $\varphi$,
since the length of $\Psi^\varphi_\cH(\bfx)$ is already doubly exponential in the length of $\varphi$, 
as the number of variables $y_{\pi, \tuple{g, \bff}}$ is doubly exponential.
Moreover, there are doubly exponentially many possible cores $\cH$.

\myparagraph{Consequences for finite satisfiability of $\Ctwo$ with global constraints}
Note that in Lemma~\ref{thm:c2_spec},
we show that for every $\Ctwo$ sentence $\varphi$ in normal form,
there exists an existential Presburger formula $\Psi^\varphi(\bfx)$ of doubly exponential size 
such that $\pispec{\varphi} = \sol{\Psi^\varphi(\bfx)}$.
Since for every 1-type $\pi$ and partial model $\cH$ of $\varphi$, we can compute $\cB^{\varphi}_{\pi, \cH}$ in doubly exponential time in the length of $\varphi$ and $\cH$ by a straightforward enumerate-and-check procedure,
$\Psi^\varphi(\bfx)$ can be computed from $\varphi$ in doubly exponential time.
Note that global constraints of $\Ctwo$ sentences are linear constraints over the $1$-type cardinality vector.
Therefore the finite satisfiability problem of $\Ctwo$ with global constraints reduces
to satisfiability problem of $\Psi^\varphi(\bfx)$ together with those linear constraints.
This naive approach yields a $2\text{-}\nexp$ algorithm.

We observe that $\Psi^\varphi(\bfx)$ can be rewritten as disjunctions of doubly exponentially many integer linear systems.
Each system contains doubly exponentially many variables, but only \emph{exponentially many constraints}.
By Lemma~\ref{lem:ilp}, if such a system has a solution,
then there exists a solution in which only exponentially many variables are assigned non-zero values.

Therefore, we can decide the finite satisfiability of $\Ctwo$ with global constraints 
using the $\nexp$ procedure in Algorithm~\ref{alg:ctwosat}.

\begin{algorithm}
\begin{algorithmic}[1]
\caption{Algorithm for the finite satisfiability problem of $\Ctwo$ with global constraints}\label{alg:ctwosat}
\Procedure{SAT}{$\varphi$}
    \State Compute $\Onetypes^\varphi$ and $\twotypes^\varphi_{\pi_1, \pi_2}$
    \State Guess a exponential size partial model $\cH$ of $\varphi$
    \State Guess exponentially many non-zero variables $y'_{\pi, \tuple{g, \bff}}$
    \State Construct the system $\Psi^\varphi_{\cH}(\bfx)$ with only variables $y'_{\pi, \tuple{g, \bff}}$
    \State Check the feasibility of $\Psi^\varphi_{\cH}(\bfx)$ with global constraints
\EndProcedure
\end{algorithmic}
\end{algorithm}

We emphasize again that the argument here is completely generic, and the same comment holds to show that for \emph{any} decidability extension of Presburger arithmetic (e.g. B\"uchi Arithemetic or Semenov Arithmetic~\cite{semenovpresexp}), $\Ctwo$ with global cardinality constraints of that form is decidable, with the complexity the maximum of $\nexptime$ and the complexity of solving the corresponding constraints.

%% file: 6-gptwospectrum.tex
\section{Negative results about the spectrum of $\GPtwo$ sentences} \label{sec:gptwocounter}
The $1$-type spectrum of a $\Ctwo$ sentence is semilinear, and we showed how to calculate it in
the previous section. For $\GPtwo$, earlier in the paper we were able to characterize models by a linear system: but the solutions
to this linear system are not in one-to-one correspondence with the number of satisfiers of $1$-types.
Thus a natural question is whether the $1$-type spectrum of a $\GPtwo$ sentence is semilinear. We answer this below in
the negative.

Consider the following $\GPtwo$ formula $\varphi$ over vocabulary $\set{P_1, P_2, E}$:
\begin{equation*}
    \begin{aligned}
        \varphi\ :=\ &
        \forall x\ 
        \left(
            \left(\neg P_1(x) \land P_2(x)\right) \lor
            \left(P_1(x) \land \neg P_2(x)\right)
        \right)
        \ \ \land \ \ 
        \bigwedge_{1 \le i \le 2} \forall x\ P_i(x) \to \gamma_i(x) \\
        \gamma_1(x)\ :=\ &
        \presby{E(x, y) \land P_1(y)} - \presby{E(x, y) \land P_2(y)} = 0 \\
        \gamma_2(x)\ :=\ &
        \presby{E(y, x) \land P_1(y)} = 1.
    \end{aligned}
\end{equation*}
Let $\cG$ be a finite model of $\varphi$.
A vertex in $\cG$ is labeled with either $P_1$ or $P_2$.
Let $V_i$ be the set of vertices in $\cG$ with label $P_i$.
For every $v \in V_1$, the number of edges from $v$ to $V_1$
is the same as the number of edges from $v$ to $V_2$.
For every vertex $v \in V_2$, there exists exactly one edge from $V_1$ to $v$.
Therefore the cardinality of $V_2$ is the same as the number of edges between vertices of $V_1$,
which is bounded by $\abs{V_1}\cdot(\abs{V_1}-1)$.
Thus the 1-type spectrum of $\varphi$ is $\setc{(0, a, b, 0) \in \bbN^4}{b \le a(a-1)}$.
Hence we have the following negative result.
\begin{theorem}
    The 1-type spectrum of a $\GPtwo$ sentence is not always semilinear.
\end{theorem}

This does not resolve whether the set of cardinalities of finite models is a semilinear set, and
it leaves open whether $\GPtwo$ with global cardinality constraints is decidable.

%% file: 7-related.tex
\section{Related work} \label{sec:related}

The logic $\GPtwo$ was shown decidable in \cite{percentage,twovarpres}, while the decidability
of $\Ctwo$ with cardinality constraints is implicit in \cite{tonyeryk, ultimatelyperiodic}. Our complexity bound for $\GPtwo$ refines the analysis of the Kleene star operator in \cite{usicalp}, which in turn relies on \cite{kleenestar}. Our analysis of $\Ctwo$ with constraints obviously relies heavily on the techniques developed by Pratt-Hartmann for analysis of $\Ctwo$, in particular normal forms, the categorization of types, and the distinction between large and small number of realizers (e.g. the notion of Z-differentiated \cite{ctwocomplexity,gctwocomplexity}).

Although we borrow techniques from prior work, our approach gives a self-contained proof of decidability for $\Ctwo$ itself, and also for $\Ctwo$ with global constraints in extensions of Presburger arithmetic. An alternative approach to get bounds for $\Ctwo$ with Presburger global constraints is by reducing to decidability of $\Ctwo$.
This approach was presented by Rudolph in \cite{rudolphcardinality}, relying on ideas from \cite{baadercardinality}. Although technically speaking the results in \cite{rudolphcardinality} are presented for description logics, they could easily be applied to give the same bounds for $\Ctwo$.

In addition to contributing tight bounds for complexity for decidability of $\GPtwo$ a well as the size of the spectrum  for $\Ctwo$, we note that \cite{usicalp} is phrased in terms of graph neural networks, while from the analysis of spectra in \cite{tonyeryk,ultimatelyperiodic} it would be difficult to see how semilinear representations are derived. Thus we believe our work additionally makes the analyses of logics with arithmetic substantially more accessible.

%% file: 8-conc.tex
\section{Conclusion} \label{sec:conc}

In the paper we isolate the complexity of satisfiability for two logics combining relational
structure and  arithmetic, $\GPtwo$ and 
the logics $\Ctwo$ and $\GCtwo$ extended with cardinality constraints. In the process we refine and
simplify proof techniques for analyzing such logics.  $\GPtwo$ and $\Ctwo$ are incomparable in expressiveness. The former allows inductively forming open formulas that involve arithmetic between counts of neighbors, while in the latter arithmetic can only be used at top-level for forming sentences. On the other hand, $\Ctwo$ allows unguarded formulas. See the appendix for a formal argument.
We leave open the question of the decidability of combinations of formulas in the two languages.

%% file: appendix.tex
\input{app-linearalgebra}
\input{app-normalform}
\input{app-lcore}
\input{app-mainctwoproof}
\input{app-modulus}

\section*{Expressiveness separation of the considered languages}
We mentioned that the languages we consider in this paper -- $\GPtwo$ and $\Ctwo$ with global cardinality constraints -- are incomparable in expressiveness.

$\GPtwo$  contains sentences not in $\Ctwo$ with constraints, because (as we showed in the paper) the spectrum of GP2 may not be semilinear, while for $\Ctwo$ with constraints the specctrum is always semilinear.

On the other hand, even ordinary $\Ctwo$ contains sentences not in $\GPtwo$: just consider the sentence restricting the cardinality of the model to some fixed number, like $2$. In $\GPtwo$ every sentence that is consistent has models of unbounded finite cardinality, since we can always clone a connected component without disturbing the truth value of any sentence.

%% file: app-linearalgebra.tex
\section{Proof of Lemma \ref{lem:key}: computation of a representation of the Kleene closure}

Recall that one important component in the decidability of $\GPtwo$ is the computation of a representation of the Kleene closure of a semilinear set.
Recall from Section \ref{sec:prelims} that $\projp{n}{\bfv}$ denotes the projection operator that maps a vector $\bfv$ to its first $n$ entries.
The formal statement was:

\medskip

For every integer linear system $\cQ(\bfx): \bfA \bfx = \bfc$ with $\bfc \neq \bfzero$,
where $\bfA \in \bbZ^{m \times n}$ and $\bfc \in \bbZ^{m}$,
there exists $\bfAtilde \in \bbZ^{(n+t) \times (n+k)}$ with $\norm{\bfAtilde} = 1$
such that
\begin{equation*}
    \kstar{\sol{\cQ(\bfx)}}\ =\ \sol{\cQtilde(\bfx)},
\end{equation*}
where $\cQtilde(\bfx)$ is the following Boolean combination of integer linear systems
\begin{equation*}    
    \exists \bfy_1, \bfy_2\ 
    (\bfx\ =\ \projp{n}{\bfy_1} + \bfy_2)\ \land\ 
    (\bfAtilde\bfy_1\ =\ \bfzero)\ \land\ 
    (\bfA\bfy_2\ =\ \bfzero)\ \land\ 
    ((\bfy_1=\bfzero) \to (\bfy_2=\bfzero)),
\end{equation*}
$k = n\cdot(2D+1)^{m} + 1$,
$t \le 2nk$, and
$D = n \cdot \norm{\bfA} \cdot \left((n+1)\cdot \norm{\bfA} + \norm{\bfc} + 1\right)^{m}$.
Moreover, $\bfAtilde$ can be computed in time $2^\bigO{\log n + \log K + m^2}$,
where $K := \max\left(\norm{\bfA}, \norm{\bfc}\right)$.

\medskip

We prove the lemma by considering \emph{minimal solutions} of the integer linear system.
Recall that $\vle$ means pointwise inequality between vectors.
\begin{definition}[Minimal elements]
    For a set $\cS \subseteq \bbN^n$ and $\bfv \in \cS$,
    we say that $\bfv$ is \emph{minimal} in $\cS$
    if, for every $\bfu \in \cS$ with $\bfu \neq \bfv$,
     $\bfu \not\vle \bfv$.
\end{definition}

Note that the following properties about minimal elements are trivial.
\begin{lemma}\label{lem:min_property}
    Let $\cS$ be a nonempty subset of $\bbN^n$.
    \begin{enumerate}
        \item There exists $\bfv \in \cS$ such that $\bfv$ is minimal in $\cS$.
        \item For every $\bfv \in \cS$, there exists a minimal element $\bfu \in \cS$,
            such that $\bfv - \bfu \in \bbN^n$.
    \end{enumerate}
\end{lemma}

For an integer linear system $\cQ(\bfx)$,
let $\solmin{\cQ(\bfx)}$ denotes the set of minimal solutions,
i.e., $\solmin{\cQ(\bfx)} := \setc{\bfv \in \sol{\cQ(\bfx)}}{\bfv\ \text{is minimal in}\ \sol{\cQ(\bfx)}}$.
Note that the maximal entry of minimal solutions is bounded.

\begin{lemma}[\cite{Pottier}, also see Proposition 4. in~\cite{semilinear}]\label{lem:minimal}
    For every integer linear system $\cQ(\bfx): \bfA \bfx = \bfc$,
    where $\bfA \in \bbZ^{m \times n}$ and $\bfc \in \bbZ^{m}$,
    for every $\bfv \in \solmin{\cQ(\bfx)}$,
     $\norm{\bfv} \le \left((n+1)\cdot\norm{\bfA} + \norm{\bfc} + 1\right)^{m}$.
\end{lemma}

We begin by turning to the automata-theoretic approach for the solutions of integer linear systems.

\begin{definition}[Acyclic, unambiguous, and simple finite automata] \label{def:simpleplus}
As usual, a finite automaton can be viewed as a 
directed graph where the vertices are the states
and the edges are the transitions.

Such an automaton is \emph{acyclic} if the graph does not contain a cycle. 
It is \emph{unambiguous} if every word has 
at most one accepting run.
    We say that a finite automaton is \emph{simple}
    if it is unambiguous, acyclic, and has exactly one accepting state.
\end{definition}

\begin{definition}[Parikh Vector]
    For a word $w = a_{t_1}\cdots a_{t_\ell} \in \Sigma^*$,
    the \emph{Parikh vector} of $w$
    is a vector in $\bbN^{n}$, defined as
    $\parikh{w} := \sum_{i \in \intsinterval{\ell}} \unit{t_i}$.
    Intuitively, the $i^{th}$ entry in $\parikh{w}$ is the number of occurrences of $a_i$ in $w$.

For a finite automaton $\cM$, the \emph{Parikh image} of $\cM$ is defined as
\begin{equation*}
\parikh{\cM}
\ :=\ 
\setc{\parikh{w}}{\text{$\cM$ accepts $w$}}.
\end{equation*}
\end{definition}

We can encode a finite subset of solutions of a linear system $\cQ(\bfx):\bfA \bfx = \bfc$ as the Parikh image of a finite automaton $\cM$.
More precisely, $\parikh{\cM}$ contains all the minimal solutions of $\cQ(\bfx)$.
The main idea behind the construction of $\cM$ is as follows.
Suppose $\cQ(\bfx)$ has a solution
and $\bfA\in \bbZ^{m\times n}$.
By Lemma~\ref{lem:minimal},
the components of every minimal solution are bounded above by  
$L:=\left((n+1)\cdot\norm{\bfA} + \norm{\bfc} + 1\right)^{m}$.
We can construct a finite automaton $\cM$ where the set of states is $\setc{\bfA\bfv}{\bfv \in [0,L]^n}$
and the transitions are
$\langle \bfA\bfv,a_i, \bfA(\bfv+\bfe_i) \rangle$ for every $1\le i \le n$.
The initial state is $\bfzero$ and the accepting state is $\bfc$.
Obviously, $\parikh{\cM}$ contains all the minimal solutions. In fact, it contains all solutions that are bounded by $L$. 
Note also that
$\cM$ can be constructed in exponential time.
With a little more work, we can make $\cM$ simple, a property that will be useful when analyzing the Kleene star of $\sol{\cQ}$.

\begin{lemma}\label{lem:ilp_to_nfa}
For every integer linear system $\cQ(\bfx): \bfA \bfx = \bfc$,
where $\bfA \in \bbZ^{m \times n}$ and $\bfc \in \bbZ^{m}$,
there exists a simple finite automaton $\cM$
with $k$ states and $t$ transitions,
where $k = n\cdot L\cdot(2D+1)^{m} + 1$, 
$t \le 2nk$,
$L = \left((n+1)\cdot\norm{\bfA} + \norm{\bfc} + 1\right)^{m}$
and $D = n \cdot \norm{\bfA} \cdot L$,
    such that
    \begin{equation*}
        \solmin{\cQ(\bfx)}
        \ \subseteq\ 
        \parikh{\cM}
        \ \subseteq\ 
        \sol{\cQ(\bfx)}.
    \end{equation*}
    Moreover, $\cM$ can be constructed in time exponential in $n,m,\norm{\bfA},\norm{\bfc}$.
\end{lemma}

\begin{proof}
If $\bfc = \bfzero$, then the system $\cQ(\bfx)$ admits the
solution $\bfzero$.
The finite automaton that accepts only the empty word is the desired $\cM$.
So we concentrate on the case when $\bfc\neq \bfzero$.

The automaton $\cM$ is very similar to the one described in the intuition preceding the lemma.
We fix a finite alphabet of $n$ symbols $\Sigma := \set{a_1, \ldots, a_{n}}$.
To ensure acyclicity, we modify the automaton $\cM$ to accept only words from $a_1^*a_2^*\cdots a_n^*$. The automaton will also implement a counter that counts the number of occurrences of each $a_i$ up to $L$. Such a counter can be implemented using the states of $\cM$.
We now give the details.

Let $\cD$ be the set $\intinterval{-D}{D}^{m}$.
\begin{itemize}
\item 
The set of states is 
$\{q_{acc}\}\cup(\intsinterval{n}\times\intsinterval{L} \times \cD)$.
\item
The initial state is $(0,0,\bfzero)$.
\item 
For each $0\leq i<j \leq n$, $0\leq t\leq L-1$,
and $\bfv \in \cD$,
we have the transitions:
\begin{itemize}
\item 
if $\bfv + \bfA\unit{i} \in \cD$,
then $(\tuple{i,t, \bfv}, a_i, \tuple{i,t+1, \bfv + \bfA\unit{i}}) \in \Delta$, 
\item 
if $\bfv + \bfA\unit{i}  = \bfc$,
then $(\tuple{i, t,\bfv}, a_i, q_{acc}) \in \Delta$.

\item 
if $\bfv + \bfA\unit{j} \in \cD$,
then $(\tuple{i,t, \bfv}, a_j, \tuple{j,1, \bfv + \bfA\unit{j}}) \in \Delta$, 
\item 
if $\bfv + \bfA\unit{j}  = \bfc$,
then $(\tuple{i, t,\bfv}, a_j, q_{acc}) \in \Delta$.
\end{itemize}
\end{itemize}
Obviously, $\cM$ has only one accepting state.
The automaton $\cM$ is acyclic since it can only go from the state
$\tuple{i,t, \bfv}$ to $\tuple{j,t',\bfu}$
when $i< j$ or when $i=j$ and $t<t'$.
It is also unambiguous since
for every state $\tuple{i,t,\bfv}$ and $a\in \Sigma$, there is at most one state $\tuple{i',t',\bfu}$
such that $\tuple{i,t,\bfv},a,\tuple{i',t',\bfu}$ is a transition.
Finally, it is routine to verify that
$\parikh{\cM}\subseteq \sol{\cQ}$.
That $\solmin{\cQ}\subseteq \parikh{\cM}$
follows from Lemma~\ref{lem:minimal}.
\end{proof}

Next, we consider the Kleene star of sets.
Recall the definition:

\medskip

For every $\cS \subseteq \bbN^n$,
the \emph{Kleene star} of $\cS$ is defined as
\begin{equation*}
    \kstar{\cS} := \setc{\sum_{s \in \cS'} s}{\text{$\cS'$ is a finite multisubset of $\cS$}}.
\end{equation*}

\medskip

In the following we fix a finite alphabet of $n$ symbols $\Sigma:=\{a_1,\ldots,a_n\}$
and assume that every finite automaton is over this alphabet.
We need more terminology concerning finite automata.
Note that for every acyclic finite automaton $\cM$ (Definition \ref{def:simpleplus})
a run $r$ of $\cM$ uses each transition at most once.
Therefore, we can represent $r$ by a Boolean vector with dimension $\abs{\Delta}$.
Recall that for a unambiguous finite automaton $\cM$ and an accepted word $w$, there exists exactly one accepting run of $\cM$ on $w$.
Let $\bbB:=\{0,1\}$.

\begin{definition}
    For every simple finite automaton $\cM$ with $t$ transitions,
    for every accepted word $w$ of $\cM$,
    letting $r$ be the unique accepting run of $\cM$ on $w$,
    the \emph{path vector} of $w$, denoted by $\mpath{w}$,
    is the vector in $\bbB^t$ defined as follows:
    for $1 \le i \le t$, the $i^{th}$ entry of $\mpath{w}$ is 1
    if and only if $r$ contains the $i^{th}$ transition.
\end{definition}

\begin{lemma}\label{lem:path}
For every simple finite automaton $\cM$ with $t$ transitions,
there exists a matrix $\bfB_\cM \in \bbZ^{n \times t}$,
which can be computed in time $\bigO{nt}$,
such that
\begin{itemize}
\item $\norm{\bfB_\cM} = 1$, and
\item for every word $w$ accepted by $\cM$, $\parikh{w} = \bfB_\cM\cdot\mpath{w}$.
\end{itemize}
\end{lemma}

\begin{proof}
    We construct the matrix $\bfB$ as follows.
    For $1 \le i \le n$ and $1 \le j \le t$,
    if the $t^{th}$ transition is of the form $(q, a_i, p)$,
    then $(B_\cM)_{i, j} = 1$.
    Otherwise, $(B_\cM)_{i, j} = 0$.
    It is trivial to verify that
    $\parikh{w} = \bfB\cdot\mpath{w}$.
\end{proof}

We now prove a  lemma regarding the Kleene closure of the Parikh image.
The proof employs a similar idea as in the argument for Theorem~1 in \cite{tree}.

\begin{lemma}\label{lem:kstar_nfa}
    For every simple finite automaton $\cM$ with $k$ states $q_1, \ldots,q_k$ and $t$ transitions,
    there exists a matrix $\bfA_\cM \in \bbZ^{k \times t}$,
    which can be computed in time $\bigO{kt}$,
    such that
    \begin{itemize}
        \item $\norm{\bfA_\cM} = 1$, and
        \item $\kstar{\parikh{\cM}} =
            \sol{\cQ_\cM(\bfx)}$,
            where
            $\cQ_\cM(\bfx) := \exists \bfy\ (\bfx = \projp{n}{\bfy}) \land (\tilde{\bfA}_\cM \bfy = \bfzero)$ and
            $\bfAtilde_\cM = 
            \left[
                \begin{array}{r|r}
                    -\bfI & \bfB_\cM \\
                    \hline
                    \bfzero& \bfA_\cM
                \end{array}
            \right]$.
    \end{itemize}
\end{lemma}

\begin{proof}
    We construct the matrix $\bfA_\cM$ as follows.
    For $1 \le i \le k$ and $1 \le j \le t$,
    if $q_i$ is the initial state or the accepting state, then
    $\left(A_\cM\right)_{i, j} = 0$.
    Otherwise,
    \begin{equation*}
        \left(A_\cM\right)_{i, j}\ :=\ 
        \begin{cases}
            1,  &\text{if the $j^{th}$ transition is of the form $(p, a, q_i)$} \\
            -1, &\text{if the $j^{th}$ transition is of the form $(q_i, a, p)$} \\
            0,  &\text{otherwise}.
        \end{cases}
    \end{equation*}
    Note that since $\cM$ is acyclic, there is no transition of the form $(q_i, a, q_i)$.

    We first claim that for every accepted word $w$, 
    we have
    \begin{equation*}
        \bfA_\cM \cdot \mpath{w} = \bfzero.    
    \end{equation*}
    We will show that $\left(\bfA_\cM \cdot \mpath{w}\right)_i = 0$,    for $1 \le i \le k$.
    If $q_i$ is the initial state or the accepting state,
    by definition, $\left(A_\cM\right)_{i, j} = 0$ for every $1 \le j \le k$,
    hence, $\left(\bfA_\cM \cdot \mpath{w}\right)_i = 0$.

    If $q_i$ does not appear in the accepting run of $w$,
    then $\mpath{w}_j=0$ if
    the $j^{th}$ transition is going into or coming out of state $q_i$.
    Hence, $\left(\bfA_\cM \cdot \mpath{w}\right)_i = 0$.

    Now, suppose that $q_i$ appears in $\mpath{w}$
    and it is 
    neither the initial state nor the accepting state.
    Let $r$ be the accepting run of $\cM$ on $w$.
    Since $\cM$ is acyclic, $q_i$ appears only once, which implies that $r$ is of the form
    \begin{equation*}
        \cdots p\transs{a} q_i \transs{a'}p'\cdots
    \end{equation*}
    Suppose that
    $\tuple{p, a, q_i}$ is the $j_1^{th}$ transition and 
    $\tuple{q_i, a', p'}$ is the $j_2^{th}$ transition.
    Then we have:
    \begin{equation*}
        \left(\bfA_\cM \cdot \mpath{w}\right)_i
        =
        \left(\bfA_\cM\right)_{i, j_1} +
        \left(\bfA_\cM\right)_{i, j_2}
        = 1 + (-1)
        = 0.
    \end{equation*}

    We will now show that
    $\kstar{\parikh{\cM}} = \sol{\cQ_\cM(\bfx)}$.

    \textbf{\uline{$\kstar{\parikh{\cM}} \subseteq \sol{\cQ_\cM(\bfx)}$}.}
    Suppose that $\bfv \in \kstar{\parikh{\cM}}$.
    By the definition of Kleene star and Parikh image,
    there exist accepted words $w_1, \ldots, w_\ell$
    and $\bfv = \sum_{i \in \intsinterval{\ell}} \parikh{w_i}$.
    Let $\bfu := \sum_{i \in \intsinterval{\ell}} \mpath{w_i}$.
    By the claim above and Lemma~\ref{lem:path}, we have
    \begin{equation*}
        \begin{aligned}
            \bfA_\cM \bfu\ =\ &\sum_{i \in \intsinterval{\ell}} \bfA \cdot \mpath{w_i}\ =\ \bfzero \\
            \bfB_\cM \bfu\ =\ &\sum_{i \in \intsinterval{\ell}} \bfB \cdot \mpath{w_i}
            \ =\ \sum_{i \in \intsinterval{\ell}} \parikh{w_i}\ =\ \bfv,
        \end{aligned}
    \end{equation*}
    which implies that concatenation of $\bfv$ and $\bfu$ is a solution of the integer linear system $\cQ_\cM(\bfx)$.
    Hence, $\bfv \in \sol{\cQ_\cM(\bfx)}$.

    \textbf{\uline{$\sol{\cQ_\cM(\bfx)}\subseteq\kstar{\parikh{\cM}}$}.}
    Suppose that $\bfv \in \sol{\cQ_\cM(\bfx)}$.
    Let $\bfu$ be the vector satisfying $\bfv = \bfB_\cM\bfu$ and $\bfA_\cM\bfu = \bfzero$.
    We interpret the $j^{th}$ entry of $\bfu$, denoted by $u_j$
    as a flow of the $j^{th}$ edge (transition) in the transition diagram of $\cM$.
    Then, since $\bfA_\cM\bfu = \bfzero$ by the claim above,
    for every node (state) $q$,
    if $q$ is neither the initial state nor the accepting state,
    then 
    \begin{equation*}
        \sum_{\substack{
                j \in \intsinterval{t}\\
                \text{the $j^{th}$ transition is of the form $(p, a, q)$}
        }} u_j
        -
        \sum_{\substack{
                j \in \intsinterval{t}\\
                \text{the $j^{th}$ transition is of the form $(q, a, p)$}
        }} u_j
        = 0.
    \end{equation*}
    That is, the incoming flow to the node (state) $q$ equals its outgoing flow.
    By the flow decomposition theorem~\cite{flow},
    there exist paths (words) $w_1, \ldots, w_\ell$ from the initial state to the accepting state
    such that $\bfu = \sum_{i \in \intsinterval{\ell}} \mpath{w_i}$.
    Since each $w_i$ is a path from the initial state to the accepting state, each
    $w_i$ is a word accepted by $\cM$.
    By Lemma~\ref{lem:path}, for each $i \leq \ell$, $\parikh{w_i} = \bfB_\cM\cdot\mpath{w_i}$.
    Thus, we have
    \begin{equation*}
        \bfv
        \ =\ \bfB_\cM\bfu
        \ =\ \sum_{i \in \intsinterval{\ell}} \bfB_\cM\cdot\mpath{w_i}
        \ =\ \sum_{i \in \intsinterval{\ell}} \parikh{w_i}
        \ \in\ \kstar{\parikh{\cM}}.
    \end{equation*}
\end{proof}

Combining Lemma~\ref{lem:ilp_to_nfa} and Lemma~\ref{lem:kstar_nfa} we obtain
a lemma regarding the Kleene star of the set of minimal solutions of an integer linear system.
\begin{lemma}\label{lemma:kleene_minimal}
    For every integer linear system $\cQ(\bfx): \bfA \bfx = \bfc$,
    where $\bfA \in \bbZ^{m \times n}$, $\bfc \in \bbZ^{m}$,
    there exists $\bfAtilde \in \bbZ^{(n+t) \times (n+k)}$ with $\norm{\bfAtilde} = 1$
    such that
    \begin{equation*}
        \kstar{\solmin{\cQ(\bfx)}}
        \ \subseteq\ 
        \sol{\exists \bfy\ (\bfx = \projp{n}{\bfy}) \land (\tilde{\bfA} \bfy = \bfzero)}
        \ \subseteq\ 
        \kstar{\sol{\cQ(\bfx)}},
    \end{equation*}
    where 
    $k = n\cdot(2D+1)^{m} + 1$,
    $t \le 2nk$, and
    $D = n \cdot \norm{\bfA} \cdot \left((n+1)\cdot \norm{\bfA} + \norm{\bfc} + 1\right)^{m}$.
    Moreover, $\bfAtilde$ can be computed in time $2^\bigO{\log n + \log K + m^2}$,
    where $K := \max\left(\norm{\bfA}, \norm{\bfc}\right)$.
\end{lemma}

We are now ready to prove Lemma \ref{lem:key}.
\begin{proof}
    Let $\cQ(\bfx): \bfA\bfx = \bfc$.
    Let $\bfAtilde$ be the matrix obtained by applying Lemma~\ref{lemma:kleene_minimal} to the integer linear system $\cQ(\bfx)$.
    It is routine to verify that the dimension and maximal absolute values of the entries of $\bfAtilde$ satisfy the required bounds, and $\bfAtilde$ can be computed within the desired time complexity.

    Here we show that $\kstar{\sol{\cQ(\bfx)}} = \sol{\cQtilde(\bfx)}$, where
    \begin{equation*}
        \cQtilde(\bfx)\ :=\ 
        \exists \bfy_1, \bfy_2\ 
        (\bfx\ =\ \projp{n}{\bfy_1} + \bfy_2)\ \land\ 
        (\bfAtilde\bfy_1\ =\ \bfzero)\ \land\ 
        (\bfA\bfy_2\ =\ \bfzero)\ \land\ 
        ((\bfy_1=\bfzero) \to (\bfy_2=\bfzero)).
    \end{equation*}

    \uline{$\kstar{\sol{\cQ(\bfx)}} \subseteq \sol{\cQtilde(\bfx)}$}.
    Suppose that $\bfv \in \kstar{\sol{\cQ(\bfx)}}$.
    If $\bfv = \bfzero$, then we can choose $\bfzero$
    as the assignment for $\bfy_1$ and $\bfy_1$.
    Otherwise if $\bfv \neq \bfzero$,
    by definition of Kleene star and minimal elements,
    there exists $\bfu_1, \bfu_2, \ldots, \bfu_k \in \sol{\cQ(\bfx)}$ and
    $\bfu_{1, \min}, \bfu_{2, \min}, \ldots, \bfu_{k, \min} \in \solmin{\cQ(\bfx)}$
    such that $\bfv = \sum_{i \in \intsinterval{k}} \bfu_i$,
    and for $1 \le i \le k$, $(\bfu_i - \bfu_{i, \min}) \in \bbN^n$.
    \begin{itemize}
        \item
            Because $\sum_{i \in \intsinterval{k}} \bfu_{i, \min} \in \kstar{\solmin{\cQ(\bfx)}}$,
            by Lemma~\ref{lemma:kleene_minimal},
            there exists
            $\bfv_1$ such that
            $\sum_{i \in \intsinterval{k}} \bfu_{i, \min} = \projp{n}{\bfv_1}$ and
            $\bfAtilde\bfv_1 = \bfzero$.
        \item
            Let $\bfv_2 = \sum_{i \in \intsinterval{k}} (\bfu_{i} - \bfu_{i, \min}$).
            Observe that
            $\bfA \bfv_2
            = \sum_{i \in \intsinterval{k}} \left(\bfA \bfu_{i} - \bfA \bfu_{i, \min}\right)
            = \bfzero$.
        \item
            Because $\bfc \neq \bfzero$, $\bfzero$ is not a solution of $\cQ(\bfx)$,
            we have $\bfv_1 \neq \bfzero$.
            Hence  $(\bfv_1=\bfzero) \to (\bfv_2=\bfzero)$ holds.
    \end{itemize}
    Since $\bfv = \projp{n}{\bfv_1} + \bfv_2$,
    $\bfv$ is a solution of $\cQtilde$, where the solution
    assigns $\bfy_1$ to $\bfv_1$ and assigns $\bfy_2$ to $\bfv_2$.

    \uline{$\sol{\cQtilde(\bfx)} \subseteq \kstar{\sol{\cQ(\bfx)}}$}.
    Suppose that $\bfv$ is a solution of $\cQtilde(\bfx)$
    by assigning $\bfv_1$  to $\bfy_1$ and assigning $\bfv_2$ to $\bfy_2$.
    \begin{itemize}
        \item
            If $\bfv_1 = \bfzero$, then $\bfv_2 = \bfzero$,
            which implies that
            $\bfv = \projp{n}{\bfv_1} + \bfv_2 = \bfzero$.
            By definition of Kleene star,
            $\bfzero \in \kstar{\sol{\cQ(\bfx)}}$.
        \item
            Otherwise if $\bfv_1 \neq \bfzero$,
            by Lemma~\ref{lemma:kleene_minimal},
            there exists $\bfu_1, \bfu_2, \ldots, \bfu_k \in \sol{\cQ(\bfx)}$ such that
            $\projp{n}{\bfv_1} = \sum_{i \in \intsinterval{k}} \bfu_i$.
            Note that because $\bfA\bfv_2 = \bfzero$, we have
            \begin{equation*}
                \bfA (\bfu_1 + \bfv_2)
                \ =\ \bfA\bfu_1 + \bfA\bfv_2
                \ =\ \bfc + \bfzero
                \ =\ \bfc,
            \end{equation*}
            which implies that $\bfu_1 + \bfv_2$ is a solution of $\cQ(\bfx)$.
            Therefore, we infer
            \begin{equation*}
                \bfv_1 + \bfv_2
                \ =\ (\bfu_1 + \bfv_2) + \sum_{i \in \intinterval{2}{k}} \bfu_i
                \ \in\ \kstar{\sol{\cQ(\bfx)}}.
            \end{equation*}
            The membership in $\kstar{\sol{\cQ(\bfx)}}$
            follows from the fact that $\bfu_1 + \bfv_2,\bfu_2,\ldots,\bfu_k$ are all solutions of $\cQ$.
    \end{itemize}
\end{proof}

%% file: app-normalform.tex
\section{From normal form sentences to general sentences} \label{app:normal}

In the main body all our results are proven for sentences in normal form.
In this appendix we will show that they also hold for general sentences.
In particular we will sketch the algorithms that transform an arbitrary ($\GPtwo$/$\Ctwo$) sentence into another sentence in the normal form and recover the properties of the original input sentence from the constructed normal form sentence.
We emphasize that the transformations we provide are a routine adaptation of the textbook transformation to Scott normal form: one can find similar results 
in~\cite{twovarpres}, for the $\GPtwo$ case,
and in~\cite{ctwocomplexity} for the $\Ctwo$ case.

We start with $\GPtwo$ sentences, showing that we can find an equisatisfiable formula in normal form.
Note that in the normal form we allow double universal quantifiers of the form:
$$
\forall x \forall y\ \left(R(x,y) \land x \neq y\right) \to \alpha(x, y)
$$
which can be considered as a shorthand for the  sentence:
$$
\forall x 
\left(
\presby{R(x, y) \land x \neq y\land \neg \alpha(x,y)}
\ =\ 0 
\right)
$$

Recall that the normal form of a $\GPtwo$ sentence has the shape:
\begin{equation*}
    \varphi\ :=\ 
    \forall x\ \gamma(x)\ \land\ 
    \bigwedge_{i \in \intsinterval{m}} \forall x \forall y\ \left(R_i(x,y) \land x \neq y\right) \to \alpha_i(x, y)\ \land\ 
    \bigwedge_{i \in \intsinterval{n}} \forall x\ U_i(x) \to P_i(x),
\end{equation*}
where
\begin{compactitem}
    \item $\gamma(x)$ is a quantifier-free formula,
    \item each $\alpha_i(x, y)$ is a quantifier-free formula,
    \item each $P_i(x)$ is a Presburger quantified formula of the form: 
    \begin{equation*}
        P_i(x)\ :=\ \left(
        \sum_{t \in \intsinterval{m}} \lambda_{i, t} \cdot \presby{R_t(x, y) \land x \neq y}
        \ \circledast_i\ \delta_i \right).
    \end{equation*}
\end{compactitem}

The main idea is as follows.
Let $\psi$ be an arbitrary $\GPtwo$ sentence.
By pushing every negation inward, we may assume that it is in negation normal form,
i.e., every negation is applied only on atomic formulas.
We then repeatedly replace every subformula $\nu(x)$ of one free variable $x$ with a fresh unary predicate $U(x)$, and conjoin the original formula with the ``definition'' of $U(x)$:
$$
\forall x \ U(x) \to \nu(x)
$$
Intuitively what this means is that
the subformula $\nu(x)$ is renamed as a unary predicate $U(x)$.
The algorithm starts bottom-up starting from the subformula $\nu(x)$
with the lowest quantifier rank.
It can be shown that the number of renaming required to obtain the normal form is linear in the length of the original sentence $\psi$~\cite[Lemma 3.3]{twovarpres} and the whole transformation runs in linear time.

Note that unguarded unary quantification is allowed in $\GPtwo$.
For universally-quantified formulas, $\forall x\ U(x)$ is already captured by the normal form.
For existential-quantified $\psi := \exists x\ U(x)$,
we introduce fresh unary predicates $U'$ and binary predicate $R_U$, and replace $\psi$ by 
\begin{equation*}
    \left(
        \forall x\ U(x) \lor U'(x)
    \right)
    \ \land\ 
    \left(
        \forall x\ U'(x) \to \presby{R_U(x, y) \land U(y)} = 1
    \right).
\end{equation*}

Note also that we can recover the model of the original sentence from the model of the constructed normal form sentence.
We simply project out the fresh unary predicates introduced when constructing the normal form sentence. 

Now, we turn to $\Ctwo$ sentences, and recall from the body of the paper the normal form we utilized:
\begin{equation*}
    \varphi\ :=\ 
    \forall x\ \gamma(x)\ \land\ 
    \forall x\ \forall y\ \left(x \neq y \to \alpha(x, y)\right)\ \land\ 
    \bigwedge_{i \in \intsinterval{m'}} \forall x\ \exists^{=k_i} y\ \left(R_i(x,y) \land x \neq y\right),
\end{equation*}
where $\gamma(x)$ and $\alpha(x, y)$ are quantifier-free formulas,
$R_i$ are binary predicates,
and $k_i \in \bbN$.
 
The main idea of the transformation is similar to the $\GPtwo$ case:
Rename every subformula $\nu(x)$ of one free variable $x$ with a fresh unary predicate $U(x)$ and conjunct the original sentence with $\forall x\  U(x)\to \nu(x)$.
It was shown in~\cite{ctwocomplexity,ianrevisited} that there is a linear time algorithm that transforms
an arbitrary $\Ctwo$ sentence $\varphi$ over vocabulary $\vocab$ to a sentence $\varphi^*$ in normal form (with additional unary predicates $V_1,\ldots,V_p$) such that:
\begin{compactitem}
\item 
For every model $\cA\models \varphi$ with domain size at least $1+\max_{i\in[m']}k_i$,
there is a model $\cA^*\models \varphi^*$, where $\cA^*$ is an extension of $\cA$ to the vocabulary $\vocab^*:=\vocab\cup\{V_1,\ldots,V_p\}$ with
$A=A^*$ and $S^{\cA}=S^{\cA^*}$, for every relation symbol $S\in\vocab$.
\item
Conversely,
for every model $\cA^*\models \varphi^*$ with domain size at least $1+\max_{i\in[m']}k_i$,
there is a model $\cA\models \varphi$, where $\cA$ is obtained by projecting out the predicates $V_1,\ldots,V_p$ in $\cA^*$.
\end{compactitem}
Using these two properties, we can generalise Theorem~\ref{thm:spectrasemilinear} to arbitrary $\Ctwo$ sentences.
First, all the sizes of the models up to 
$1+\max_{i\in [m']} k_i$ can be encoded inside the corresponding Presburger formula $\Psi^\varphi(\bfx)$
mentioned in Lemma~\ref{thm:c2_spec}.
Second, we can recover $\pispec$ from $\Pi^*\textsf{-SPEC}$,
where $\Pi$ and $\Pi^*$ are the set of all 1-types w.r.t. $\vocab$ and $\vocab^*$,
via the following identity, for every 1-type $\pi$ over $\vocab$:
$$
x_{\pi} \ = \
\sum_{\pi^*\supseteq \pi \ \text{and}\ \pi^*\ \text{is 1-type over}\ \vocab^*} x_{\pi^*}
$$

%% file: app-lcore.tex
\section*{Proof of the small core lemma: lemma \ref{lem:small_core}}
Recall the statement of Lemma \ref{lem:small_core}:

\medskip

    For every $\vocab$-structure $\cG$, for every $\ell \in \bbN$,
    $\cG$ has a $\ell$-core $\cH$ with size at most $(2^{2n+4m+2}) \cdot \ell$.

\medskip

\begin{proof}
    Recall from the body that the intuition behind the construction is to repeatedly select a 1-type or tuple that violates the $\ell$-core condition
    and add the relevant vertices to $\cH$.
    Since there are only finitely many such choices, the process eventually stabilizes,
    and the size of the constructed core remains bounded.

    We formalize this intuition with the following $2^{2n+4m+1}$ round procedure.
    Let $\cH_0$ be an empty $\vocab$-structure.
    For the $i^{th}$ round,
    we select either:
    \begin{compactitem}
        \item
            a 1-type $\pi$ that is realized in $\cG \setminus \cH_{i-1}$ by at most $\ell$ vertices, or

        \item
            a tuple $\tuple{\pi_1, \eta, \pi_2} \in \onetypes \times \twotypes \times \onetypes$ that is realized in $\cG \setminus \cH_{i-1}$ by at most $\ell$ pairs of vertices.
    \end{compactitem}
    We define $\cH_{i}$ as follows:
    \begin{compactitem}
        \item
           If a 1-type $\pi$ is selected, then $\cH_i$ is obtained by extending $\cH_{i-1}$
           with all vertices from $\cG \setminus \cH$ with type $\pi$. 
            Note that there are at most $\ell$ vertices added.
        \item
            If a tuple $\tuple{\pi_1, \eta, \pi_2}$ is selected,
            then $\cH_i$ is obtained by extending  $\cH_{i-1}$ with all vertices from $\cG \setminus \cH$ involved in a pair realizing $\tuple{\pi_1, \eta, \pi_2}$.
            Note that there are at most $2\ell$ vertices added.
        \item Otherwise if there is no such 1-type of tuple, $\cH_i := \cH_{i-1}$.
    \end{compactitem}
    Finally, $\cH := \cH_{2^{2n+4m+1}}$.

    Note that in the $i^{th}$ round, we add at most $2\ell$ vertices to $\cH_i$.
    Thus the size of $\cH$ is bounded by $2^{2n+4m+1} \cdot 2\ell = 2^{2n+4m+2} \cdot \ell$.
    We claim that $\cH$ is an $\ell$-core of $\cG$.
    In the $i^{th}$ round, if a 1-type $\pi$ is selected,
    then for every $j \ge i$,
    no vertex realized $\pi$ remains in $\cG \setminus \cH_j$.
    A similar property holds if a tuple is selected.
    Since there are $2^{n+m}$ possible 1-types and $2^{2n+4m}$ possible tuples,
    after $2^{2n+4m+1} \ge 2^{n+m} + 2^{2n+4m}$ rounds,
    no valid selection remains.
    Hence, $\cH$ is an $\ell$-core of $\cG$.
\end{proof}

%% file: app-mainctwoproof.tex
\section*{Proof of correctness of the main construction for $\Ctwo$: Lemma \ref{lem:c2_spec_core}} \label{app:ctwoproof}

We recall details of
the construction of a Presburger formula that is supposed to capture the spectrum of a $\Ctwo$ formula.

Recall the definitions of various collections of $2$-types:
$\cKvpnull_{\pi_1, \pi_2} \subseteq \twotypes^\varphi_{\pi_1, \pi_2}$ as the set of 2-types that are both forward- and backward-silent,
$\cKvpf_{\pi_1, \pi_2}$ for 2-types that are not forward-silent but are backward-silent.
In the appendix we also utilize the class
$\cKvpb_{\pi_1, \pi_2}$ for 2-types that are forward-silent but not backward-silent. We let
$\cKvpfb_{\pi_1, \pi_2}$ denote 2-types that are neither forward-nor backward-silent.

Recall that an extended behavior, with respect to a finite set of vertices and a $1$-type, consists of two components, where the first assigns elements in the set to a $2$-type, and the second is a behavior vector. 
Recall also the construction of the formula:

\begin{compactitem}
    \item The variables are $x_\pi$ for $\pi \in \onetypes^\varphi$,
    and $y_{\pi, \tuple{g, \bff}}$ for $\pi \in \onetypes^\varphi$ and $\tuple{g, \bff} \in \cB^\varphi_{\pi, \cH}$.
    Intuitively, $\bfx$ is the 1-type spectrum of the finite model $\cG$ and $y_{\pi, \tuple{g, \bff}}$ is the number of vertices in $\cG \setminus \cH$ with 1-type $\pi$ and extended behavior $\tuple{g, \bff}$.
    \item
    The formula $\summ_{\cH}(\bfx, \bfy)$ asserts that the 1-type spectrum of $\cG$ is the sum of the 1-type spectrum of $\cH$ and 
    the $1$-type spectrum of $\cG \setminus \cH$.
    Let $s_\pi$ be the number of vertices in $\cH$ with 1-type $\pi$. \\
    \begin{notsotiny}
    $\displaystyle\summ_{\cH}(\bfx, \bfy)
         := 
     \bigwedge_{\pi \in \onetypes^\varphi} \left(x_\pi = s_\pi + \sum_{\tuple{g, \bff} \in \cB^{\varphi}_{\pi, \cH}} y_{\pi, \tuple{g, \bff}}\right)$.
    \end{notsotiny}
    \item
    The formula $\comp_{\cH}(\bfy)$ guarantees that the vertices in $\cH$ satisfies the Counting condition of Proposition \ref{prop:ctwo_graph}.
    Let $V_c$ be the set of vertices in $\cH$. \\
    \begin{notsotiny}
    $\displaystyle\comp_{\cH}(\bfy)
         := 
        \bigwedge_{v \in V_c} \left(
            \sum_{u \in V_c \setminus \set{v}} \twotypeof^\trif(v, u) +
            \sum_{\pi \in \onetypes^\varphi} \sum_{\tuple{g, \bff} \in \cB^{\varphi}_{\pi, \cH}} g^\trib(v) \cdot y_{\pi, \tuple{g, \bff}}
        = \bfk^\varphi \right)$.
        \end{notsotiny}
    \item
    The formula $\silent(\bfy)$ assert that 1-types realized in $\cG \setminus \cH$ are silent compatible. \\
    \begin{notsotiny}
    $\displaystyle\silent_{\cH}(\bfy)
         := 
        \bigwedge_{\substack{\pi_1, \pi_2 \in \onetypes^\varphi \\ \text{are not silent compatible}}}
        \left(\sum_{\tuple{g, \bff} \in \cB^{\varphi}_{\pi_1, \cH}} y_{\pi_1, \tuple{g, \bff}} = 0\right) \lor
        \left(\sum_{\tuple{g, \bff} \in \cB^{\varphi}_{\pi_2, \cH}} y_{\pi_2, \tuple{g, \bff}} = 0\right)$.
        \end{notsotiny}
    \item
    The formula $\bige_{\cH}(\bfy)$ guarantees that 1-types and tuples realized in $\cG \setminus \cH$ satisfy the requirements of a core. \\
    \begin{notsotiny}
    $
    \begin{aligned}
        \bige_{\cH}(\bfy)
        \ :=\ &
        \bigwedge_{\pi \in \onetypes^\varphi}
            \exists z\ 
            \left(
                z = \sum_{\tuple{g, \bff} \in \cB^\varphi_{\pi, \cH}} y_{\pi, \tuple{g, \bff}}
            \right) \land
            \left(
                z = 0  \lor
                z \ge (2K^\varphi+1)^2 
            \right)
        \land \\
        &\bigwedge_{\substack{
            \pi_1, \pi_2 \in \onetypes^\varphi \\ 
            \eta \in \cKvpfb_{\pi_1, \pi_2}}}
        \exists z\ 
        \left(z = \sum_{\tuple{g, \bff} \in \cB^\varphi_{\pi_1, \cH}} \bff(\eta, \pi_2) \cdot y_{\pi_1, \tuple{g, \bff}}\right) \land
        \left(
            z = 0 \lor
            z \ge (2K^\varphi+1)^2
            \right).
        \end{aligned}
    $
    \end{notsotiny}
    \item
        The formula $\match_{1, \cH}(\bfy)$ encodes the edge matching condition for audible 2-types and vertices in $\cG \setminus \cH$. \\
        \begin{notsotiny}
        $\displaystyle
        \match_{1, \cH}(\bfy)
         := 
        \bigwedge_{\substack{
            \pi_1, \pi_2 \in \onetypes^\varphi \\
            \eta \in \cKvpfb_{\pi_1, \pi_2}}}
        \left(
            \sum_{\tuple{g, \bff} \in \cB^{\varphi}_{\pi_1, \cH}} \bff(\eta, \pi_2) \cdot y_{\pi_1, \tuple{g, \bff}}
            =
            \sum_{\tuple{g, \bff} \in \cB^{\varphi}_{\pi_2, \cH}} \bff(\dual{\eta}, \pi_1) \cdot y_{\pi_2, \tuple{g, \bff}}
        \right)
        $.
        \end{notsotiny}

    \item
        The formula $\match_{2, \cH}(\bfy)$ encodes the edge matching condition for backward-silent but not forward-silent 2-types and vertices in $\cG \setminus \cH$. \\
        \begin{notsotiny}
        $\displaystyle
        \match_{2, \cH}(\bfy)
         := 
        \bigwedge_{\substack{
            \pi_1, \pi_2 \in \onetypes^\varphi \\
            \eta \in \cKvpf_{\pi_1, \pi_2}}}
        \left(
            \left(
                \sum_{\tuple{g, \bff} \in \cB^{\varphi}_{\pi_1, \cH}} \bff(\eta, \pi_2) \cdot y_{\pi_1, \tuple{g, \bff}} > 0
            \right)
            \to
            \left(
                \sum_{\tuple{g, \bff} \in \cB^{\varphi}_{\pi_2, \cH}} y_{\pi_2, \tuple{g, \bff}} > 0
            \right)
        \right)
        $.
        \end{notsotiny}
        
        Keep in mind that we will always be considering solutions in the natural numbers. So this implication could be rewritten without summation: if for one of the extended behaviors for this triple, $\bff(\eta, \pi_2) \cdot y_{\pi_1, \tuple{g, \bff}} > 0$, then  one of the numbers $y_{\pi_2, \tuple{g, \bff}} > 0$.
\end{compactitem}
    Finally, $\displaystyle
        \Psi^\varphi_{\cH}(\bfx)
             :=
             \exists \bfy\ 
            \summ_{\cH}(\bfx, \bfy) \land
            \comp_{\cH}(\bfy) \land
            \silent_{\cH}(\bfy) \land
            \bige_{\cH}(\bfy) \land
            \bigwedge_{i=1, 2}
            \match_{i, \cH}(\bfy)$.
            
The correctness of this construction was stated in Lemma \ref{lem:c2_spec_core}:

\medskip

    For every $\Ctwo$ sentence $\varphi$ and partial model $\cH$ of $\varphi$,
    a vector $\bfv$ is a solution of $\Psi^\varphi_\cH(\bfx)$
    if and only if
    there exists a finite model $\cG$ of $\varphi$ such that
    the 1-type cardinality vector of $\cG$ is $\bfv$,
    and $\cH$ is a $(2K^\varphi+1)^2$-core of $\cG$.

\medskip

We now provide the proof.

\begin{proof}
    \textbf{\uline{There exists a finite model $\cG$ of $\varphi$ such that the 1-type cardinality vector of $\cG$ is $\bfv$ and $\cH$ is a $(2K^\varphi + 1)^2$-core of $\cG$}}.
    Let $a_{\pi}$ be the number of vertices in $\cG$ with 1-type $\pi$, and
    $b_{\pi, \tuple{g, \bff}}$ be the number of vertices in $\cG \setminus \cH$ with 1-type $\pi$ and extended behavior $\tuple{g, \bff}$.
    We claim that $a_\pi$ is a solution of $\Psi^\varphi_{\cH}(\bfx)$ with $b_{\pi, \tuple{g, \bff}}$ as the corresponding assignment to $\bfy$.
    Note that $a_\pi$ is the $\pi^{th}$ component the of 1-type cardinality vector of $\cG$,
    and
    $\sum_{\tuple{g, \bff} \in \cB^\varphi_{\pi, \cH}} b_{\pi, \tuple{g, \bff}}$ vertices in $\cG \setminus \cH$ with 1-type $\pi$.
    Let $V$ be the set of vertices in $\cG$ and $V_c \subseteq V$ be the set of vertices in $\cH$.
    \begin{itemize}
        \item ($\summ_{\cH}(\bfx, \bfy)$)
            For every $\pi$, the number of vertices with 1-type $\pi$ in $\cG$
            is the sum of the number of vertices with 1-type $\pi$ in $\cH$ and $\cG \setminus \cH$.
            Recall that $s_\pi$ is the number of vertices with 1-type $\pi$ in $\cH$.
            Thus it holds that
            \begin{equation*}
                a_\pi\ =\ s_\pi + \sum_{\tuple{g, \bff} \in \cB^\varphi_{\pi, \cH}} y_{\pi, \tuple{g, \bff}}.
            \end{equation*}

        \item ($\comp_{\cH}(\bfy)$)
        For every vertex $v$ in $\cH \subseteq \cG$,
        because $\cG$ is a model of $\varphi$, $v$ satisfies the counting condition of Proposition~\ref{prop:ctwo_graph},
        \begin{equation*}
            \sum_{u \in V \setminus \set{v}} \twotypeof^\trif(v, u)
                \ =\ 
            \bfk^\varphi.
        \end{equation*}
        Therefore, we have
        \begin{equation*}
            \begin{aligned}
                &
                \sum_{u \in V_c \setminus \set{v}} \twotypeof^\trif(v, u) +
                \sum_{\pi \in \onetypes^\varphi} \sum_{f \in \cB^\varphi_{\pi, \cH}} g(v)^\trib \cdot b_{\pi, f} \\
                \ =\ &
                \sum_{u \in V_c \setminus \set{v}} \twotypeof^\trif(v, u) +
                \sum_{u \in V \setminus V_c} [\bh{u}]^\trib(v) \\
                \ =\ &
                \sum_{u \in V_c \setminus \set{v}} \twotypeof^\trif(v, u) +
                \sum_{u \in V \setminus V_c} \twotypeof^\trib(u, v) \\
                \ =\ &
                \sum_{u \in V \setminus \set{v}} \twotypeof^\trif(v, u)
                \ =\ 
                \bfk^\varphi.
            \end{aligned}
        \end{equation*}

        \item ($\silent_{\cH}(\bfy)$)
            Because $\cH$ is a $(2K^\varphi+1)^2$-core of $\cG$,
            by Lemma~\ref{lem:mutually_null_compatible},
            each pair of 1-types $\pi_1$ and $\pi_2$ realized in $\cG \setminus \cH$ is silent compatible.
            Thus the constraint holds for $\bfb_{\pi, \tuple{g, \bff}}$.

        \item ($\bige_{\cH}(\bfy)$)
            Because $\cH$ is a $(2K^\varphi+1)^2$-core of $\cG$,
            by the definition of a core,
            each 1-type $\pi$ realized in $\cG \setminus \cH$
            is realized by at least $(2K^\varphi+1)^2$ vertices in $\cG \setminus \cH$.
            The analogous condition holds for a tuple of 1- and 2-types.
            Thus the constraint holds for $\bfb_{\pi, \tuple{g, \bff}}$.

        \item ($\match_{1, \cH}(\bfy)$)
            For every tuple $\tuple{\pi_1, \eta, \pi_2} \in \onetypes^\varphi \times \cKvpfb_{\pi_1, \pi_2} \times \onetypes^\varphi$,
            for every pair of vertices $v$ and $u$ in $V \setminus V_c$
            with $\onetypeof(v) = \pi_1$ and $\onetypeof(u) = \pi_2$,
            if $\twotypeof(v, u) = \eta$,
            then $\twotypeof(u, v) = \dual{\eta}$.
            Thus we have
            \begin{equation*}
                \sum_{\tuple{g, \bff} \in \cB^\varphi_{\pi_1, \cH}} \bff(\eta, \pi_2) \cdot b_{\pi_1, \tuple{g, \bff}}
                =
                \sum_{\tuple{g, \bff} \in \cB^\varphi_{\pi_2, \cH}} \bff(\dual{\eta}, \pi_1) \cdot b_{\pi_2, \tuple{g, \bff}}.
            \end{equation*}

        \item ($\match_{2, \cH}(\bfy)$)
            For every tuple $\tuple{\pi_1, \eta, \pi_2} \in \onetypes^\varphi \times \cKvpf_{\pi_1, \pi_2} \times \onetypes^\varphi$,
            if there exists a pair of vertices $v$ and $u$ in $V \setminus V_c$
            with $\onetypeof(v) = \pi_1$, $\onetypeof(u) = \pi_2$,
            and $\twotypeof(v, u) = \eta$,
            then
            \begin{equation*}
                \begin{aligned}
                    \sum_{\tuple{g, \bff} \in \cB^\varphi_{\pi_1, \cH}} \bff(\eta, \pi_2) \cdot b_{\pi_1, \tuple{g, \bff}}\ >\ &0 \\
                    \sum_{\tuple{g, \bff} \in \cB^\varphi_{\pi_2, \cH}} b_{\pi_2, \tuple{g, \bff}}\ >\ &0.
                \end{aligned}
            \end{equation*}
            Thus the condition holds.
            Otherwise, if there is no such pair, then
            \begin{equation*}
                 \sum_{\tuple{g, \bff} \in \cB^\varphi_{\pi_1, \cH}} \bff(\eta, \pi_2) \cdot b_{\pi_1, \tuple{g, \bff}}\ =\ 0.   
            \end{equation*}
            The condition holds trivially.
    \end{itemize}

    \begin{figure}[!ht]
        \begin{center}
            \begin{tikzpicture}
                \draw (0,0) ellipse (4cm and 1cm);
                \node at (-3, 1) {\large $\cH$}; 

                \node[circle,fill=blue,inner sep=0pt,minimum size=3pt,label=below:{\scriptsize $u_1$}] (v1) at (-3, 0) {};
                \node[circle,fill=blue,inner sep=0pt,minimum size=3pt,label=below:{\scriptsize $u_2$}] (v2) at (-1, 0) {};
                \node[circle,fill=blue,inner sep=0pt,minimum size=3pt,label=below:{\scriptsize $u_3$}] (v3) at ( 1, 0) {};
                \node at ( 2, 0) {$\ldots$}; 

                \draw (-3.5, -2) ellipse (3cm and 0.8cm);
                \node at (-6.5,-2.5) {\large $V_{\pi_1}$}; 

                \node at (-5.7,-2) {\scriptsize $V_{\pi_1, \tuple{g_1, \bff_1}}$};
                \node at (-4.2,-2) {\scriptsize $V_{\pi_1, \tuple{g_2, \bff_2}}$};
                \node at (-2.7,-2) {\scriptsize $V_{\pi_1, \tuple{g_3, \bff_3}}$};
                \node at (-1.5,-2) {$\ldots$}; 
                \draw[red] (-5,  -1.3) -- (-5,  -2.7);
                \draw[red] (-3.5,-1.2) -- (-3.5,-2.8);
                \draw[red] (-2,  -1.3) -- (-2,  -2.7);

                \node[circle,fill=blue,inner sep=0pt,minimum size=3pt,label=below:{\scriptsize $v$}] (u) at (-4.3, -1.5) {};

                \draw[gray!40] (v1) --  node[color=black,above=-2,sloped] {\scriptsize $g_2(u_1)$} (u);
                \draw[gray!40] (v2) --  node[color=black,above=-2,sloped] {\scriptsize $g_2(u_2)$} (u);
                \draw[gray!40] (v3) --  node[color=black,above=-2,sloped] {\scriptsize $g_2(u_3)$} (u);

                \draw (2, -2) ellipse (2cm and 0.8cm);
                \node at (0,-2.5) {\large $V_{\pi_2}$}; 

                \node at (5,-2) {\large $\ldots$}; 
            \end{tikzpicture}
        \end{center}
        \caption{An illustration of the model construction extending the partial model $\cH$  for the proof of Theorem~\ref{lem:c2_spec_core}.}
        \label{fig:construct}
    \end{figure}
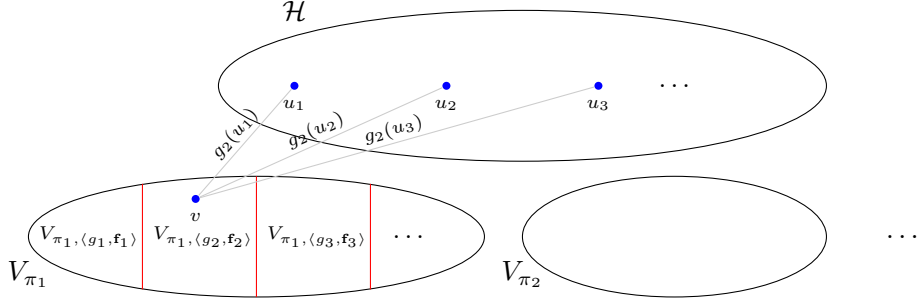

    \textbf{\uline{Suppose that $\bfv \in \sol{\Psi^\varphi_\cH(\bfx)}$}}.
    For every $\pi \in \onetypes^\varphi$ and
    $\tuple{g, \bff} \in \cB^\varphi_{\pi, \cH}$,
    let $a_\pi$ be the $\pi^{th}$ component of $\bfv$ and
    $b_{\pi, \tuple{g, \bff}}$ be the corresponding assignment to $\bfy$.
    Our goal is to construct a finite model $\cG$ of $\varphi$ by extending the partial model $\cH$.
    The construction is illustrated in Figure~\ref{fig:construct}.
    We extend the partial model $\cH$ by introducing sets of fresh vertices $V_{\pi, \tuple{g, \bff}}$,
    where the cardinality of each set is $b_{\pi, \tuple{g, \bff}}$.
    All vertices in $V_{\pi, \tuple{g, \bff}}$ are assigned the 1-type $\pi$.
    Let $V_{\pi}$ be the union of $V_{\pi, \tuple{g, \bff}}$ for every $\tuple{g, \bff}$.

    We handle edge coloring using a four-step procedure which we now outline.
    Firstly, we
    color edges between the core $\cH$ and the relative complement of the core $\cG \setminus \cH$.
    In steps two and three,
    we iterate over all pairs of 1-types $\pi_1$ and $\pi_2$ realized in the relative complement of the core $\cG \setminus \cH$ and color edges between $V_{\pi_1}$ and $V_{\pi_2}$ as follows:
    \begin{itemize}
        \item In step 2, we color edges with 2-types that are neither forward-silent nor backward-silent.
        \item In step 3, we color edges with 2-types that are not forward-silent,  but which are backward-silent.
    \end{itemize}
    We mention that the following invariant will be maintained during the iteration described above:

    \medskip

    For every $\pi \in \onetypes^\varphi$ and $\tuple{g, \bff} \in \cB^\varphi_{\pi, \cH}$,
    for every vertex $v$ in $V_{\pi, \tuple{g, \bff}}$,
    for every $\pi' \in \onetypes^\varphi$ and $\eta \in \cKvpf_{\pi, \pi'} \cup \cKvpfb_{\pi, \pi'}$,
    the number of colored edges from $v$ to $V_{\pi'}$ with 2-type $\eta$ is either 0 or $\bff(\eta, \pi')$.

    \medskip

    After the iteration, in step 4 we handle the remaining edges with 2-types that are both forward- and backward-silent.
    
    We now provide the details of the procedure.
    \begin{enumerate}
        \item 
            For every $v$ in $\cG \setminus \cH$  and $u$ in $\cH$,
            we add an edge between $v$ and $u$ with 2-type $g(u)$.
        \item
            For every $\pi_1, \pi_2 \in \onetypes^\varphi$,
            we assign audible 2-types in $\cKvpfb_{\pi_1, \pi_2}$ and $\cKvpfb_{\pi_2, \pi_1}$
            between $V_{\pi_1}$ and $V_{\pi_2}$.
            We first assume that parallel edges are allowed,
            Because $b_{\pi, \tuple{g, \bff}}$ is a solution of $\match_{1, \cH}(\bfy)$,
            it is straightforward to assign (parallel) edges between $V_{\pi_1}$ and $V_{\pi_2}$ 
            such that
            \begin{itemize}
            \item
                for every $v \in V_{\pi_1, \tuple{g, \bff}}$, for every $\eta \in \cKvpfb_{\pi_1, \pi_2}$,
                the number of edges from $v$ to $V_{\pi_2}$ with 2-type $\eta$ is exactly $\bff(\eta, \pi_2)$.
            \item
                for every $v \in V_{\pi_2, \tuple{g, \bff}}$, for every $\eta \in \cKvpfb_{\pi_2, \pi_1}$,
                the number of edges from $v$ to $V_{\pi_1}$ with 2-type $\eta$ is exactly $\bff(\eta, \pi_1)$.
            \end{itemize}

            To eliminate parallel edges in the above construction,
            we again apply a swapping procedure, as in \cite{tonyeryk,ultimatelyperiodic}.
            Note that each vertex $v \in V_{\pi_1} \cup V_{\pi_2}$ is adjacent to at most $K^\varphi$ edges.
            Suppose that there are parallel edges $e_1$ and $e_2$
            between $v_1 \in V_{\pi_1}$ and $v_2 \in V_{\pi_2}$ with the 2-type of $e_1$ being $\eta_1$ and
            the $2$-type of $e_2$ being $\eta_2$.
            Let $u_1 \in V_{\pi_1}$ and $u_2 \in V_{\pi_2}$ be vertices such that
            there is an edge $e_3$ between them with 2-type $\eta_2$,
            $u_1$ is not adjacent to $v_2$, and
            $u_2$ is not adjacent to $v_1$.
            Because $b_{\pi, \tuple{g, \bff}}$ is a solution of $\bige_{\cH}(\bfy)$,
            there are at least $(2K^\varphi+1)^2$ edges from $V_{\pi_1}$ to $V_{\pi_2}$ with 2-type $\eta_2$.
            Since each of $v_1$ and $v_2$ has at most $K^\varphi$ adjacent edges,
            it is always possible to find such $u_1$ and $u_2$.
            We then remove the edges $e_2$ and $e_3$,
            and replace them with fresh edges $(v_1, u_2)$ and $(v_2, u_1)$ with 2-type $\eta_2$.
            This swapping procedure preserves the number of edges connected to vertices with specific 2-types
            but reduces the number of parallel edges by 1,
            We can repeat this until all parallel edges are eliminated.

        \item
            For every $\pi_1, \pi_2 \in \onetypes^\varphi$,
            we assign backward-silent but not forward-silent 2-types in $\cKvpf_{\pi_1, \pi_2}$ and $\cKvpf_{\pi_2, \pi_1}$
            between $V_{\pi_1}$ and $V_{\pi_2}$.
            Again, we assume first that parallel edges are allowed,
            Because $b_{\pi, \tuple{g, \bff}}$ is a solution of $\match_{2, \cH}(\bfy)$,
            if $\sum_{\tuple{g, \bff} \in \cB^\varphi_{\pi_1, \cH}} \bff(\eta, \pi_2) \cdot b_{\pi_1, \tuple{g, \bff}} > 0$,
            then $\sum_{\tuple{g, \bff} \in \cB^\varphi_{\pi_2, \cH}} b_{\pi_2, \tuple{g, \bff}} > 0$,
            which implies that $V_{\pi_2}$ is not empty.
            Thus we can assign (parallel) edges from $V_{\pi_1}$ and $V_{\pi_2}$ satisfying that
            \begin{itemize}
            \item
                for every $v \in V_{\pi_1, \tuple{g, \bff}}$, for every $\eta \in \cKvpf_{\pi_1, \pi_2}$,
                the number of edges from $v$ to $V_{\pi_2}$ with 2-type $\eta$ is exactly $\bff(\eta, \pi_2)$.
            \end{itemize}
            A similar argument applies to edges from $V_{\pi_2}$ to $V_{\pi_1}$.
            \begin{itemize}
            \item
                for every $v \in V_{\pi_2, \tuple{g, \bff}}$, for every $\eta \in \cKvpf_{\pi_2, \pi_1}$,
                the number of edges from $v$ to $V_{\pi_1}$ with 2-type $\eta$ is exactly $\bff(\eta, \pi_1)$.
            \end{itemize}

            We apply the swapping procedure to eliminate parallel edges as well.
            Suppose that there are parallel edges $e_1$ and $e_2$
            between $v_1 \in V_{\pi_1}$ and $v_2 \in V_{\pi_2}$ with 2-types $\eta_1 \in \cKvpf_{\pi_1, \pi_2}$ for $e_1$ and $\eta_2$ for $e_2$.
            We consider two cases:
            \begin{itemize}
            \item 
                If there is a vertex $u_2 \in V_{\pi_2}$ such that there is no edges between $v_1$ and $u_2$,
                then we remove $e_1$ and add an edge between $v_1$ and $u_2$ with 2-types $\eta_1$.
            \item 
                Otherwise, suppose that there is no such $u_2$,
                which implies that for every vertex $u \in V_{\pi_2}$,
                there exists an edge $e_3$ between $v_1$ and $u$.
                Because there are at most $2K^\varphi(2K^\varphi+1)^2$ edges
                between $V_{\pi_1}$ and $V_{\pi_2}$,
                we can find $u_1 \in V_{\pi_1}$ and $u_2 \in V_{\pi_2}$
                such that no edges exists between them,
                and there exists an edge $e_3$ between $v_1$ and $u_2$ with type $\eta_3 \in \cKvpb_{\pi_1, \pi_2}$. 
                We then remove the edges $e_1$ and $e_3$
                and fresh edges $(v_1, u_2)$ with 2-type $\eta_1$ and $(v_2, u_2)$ with 2-type $\eta_3$.
            \end{itemize}
            For both cases, the swapping procedure preserves the number of edges connected to vertices with specific 2-types
            but reduces the number of parallel edges by 1,
            We can repeat it until all parallel edges are eliminated.

        \item
            Finally, for every pair of vertices $v_1 \in V_{\pi_1}$ and $v_2 \in V_{\pi_2}$ 
            with no edge connected,
            we add an edge between them with a 2-type  $\eta \in \cKvpnull_{\pi_1, \pi_2}$.
    \end{enumerate}

    We first claim that the 1-type cardinality vector of $\cG$ is exactly $\bfv$.
    For every $\pi \in \onetypes^\varphi$,
    recall that the number of vertices in $\cH$ with 1-type $\pi$ is $s_\pi$.
    The number of vertices in $\cG \setminus \cH$ with 1-type $\pi$ is
    $\abs{V_\pi} = \sum_{\tuple{g, \bff} \in \cB^\varphi_{\pi, \cH}} \abs{V_{\pi, \tuple{g, \bff}}} = \sum_{\tuple{g, \bff} \in \cB^\varphi_{\pi, \cH}} b_{\pi, \tuple{g, \bff}}$.
    Since $a_\pi$ and $b_{\pi, \tuple{g, \bff}}$ is a solution of $\summ_{\cH}(\bfx, \bfy)$,
    we have $a_\pi = s_\pi + \sum_{f \in \cB^\varphi_{\pi, \cH}} b_{\pi, \tuple{g, \bff}}$.
    Recall that $a_\pi$ is the $i^{th}$ component of $\bfv$, where $i$ is the index of $\pi$ in the  ordering of $\onetypes^\varphi$.
    Thus the 1-type cardinality vector of $\cG$ is exactly $\bfv$.

    We next claim that $\cG$ is a model of $\varphi$, by checking the conditions in Proposition~\ref{prop:ctwo_graph}.
    It is routine to check that all vertices and edges in $\cG$ satisfy the 1- and 2-type conditions.
    Thus it is sufficient to check that the counting condition holds for every vertex $v$ in $\cG$.
    \begin{itemize}
        \item
        For every $v$ in $\cH$, note that
        \begin{equation*}
            \sum_{u \in V_{\pi, \tuple{g, \bff}}} \twotypeof^\trif(v, u) 
            \ =\ 
            \sum_{u \in V_{\pi, \tuple{g, \bff}}} g^\trib(v)
            \ =\ 
            g^\trib(v) \cdot \abs{V_{\pi, \tuple{g, \bff}}}
            \ =\ 
            g^\trib(v) \cdot b_{\pi, \tuple{g, \bff}}.
        \end{equation*}
        Because $b_{\pi, \tuple{g, \bff}}$ is a solution of $\comp_{\cH}(\bfy)$, we have
        \begin{equation*}
            \begin{aligned}
                \sum_{u \in V \setminus \set{v}} \twotypeof^\trif(v, u)
                \ =\ &
                \sum_{u \in V_c \setminus \set{v}} \twotypeof^\trif(v, u) +
                \sum_{\pi \in \onetypes^\varphi}\sum_{\tuple{g, \bff} \in \cB^\varphi_{\pi, \cH}}\sum_{u \in V_{\pi, \tuple{g, \bff}}} \twotypeof^\trif(v, u) \\
                \ =\ &
                \sum_{u \in V_c \setminus \set{v}} \twotypeof^\trif(v, u) +
                \sum_{\pi \in \onetypes^\varphi}\sum_{\tuple{g, \bff} \in \cB^\varphi_{\pi, \cH}} g^\trib(v)\cdot b_{\pi, \tuple{g, \bff}} \\
                \ =\ &
                \bfk^\varphi
            \end{aligned}
        \end{equation*}

        \item
        For every $v \in V_{\pi, \tuple{g, \bff}}$,
        note that the extended behavior of $v$ is $\tuple{g, \bff}$.
        Since $\tuple{g, \bff} \in \cB^\varphi_{\pi, \cH}$,
        by the compatibility condition in Definition~\ref{def:compatextended}, it holds that
        \begin{equation*}
            \begin{aligned}
                \sum_{u \in V \setminus \set{v}} \twotypeof^\trif(v, u)
                \ =\ &
                \sum_{u \in V_c} \twotypeof^\trif(v, u) +
                \sum_{u \in V \setminus (V_c \cup \set{v})} \twotypeof^\trif(v, u) \\
                \ =\ &
                \sum_{u \in V_c} g^\trif(u) + 
                \sum_{\pi' \in \onetypes^\varphi}
                \left(
                    \sum_{\eta \in \cKvpf_{\pi, \pi'}} \bff(\eta, \pi') \cdot \eta^\trif +
                    \sum_{\eta \in \cKvpfb_{\pi, \pi'}} \bff(\eta, \pi') \cdot \eta^\trif
                \right) \\
                \ =\ &\bfk^\varphi
            \end{aligned}
        \end{equation*}
    \end{itemize}
\end{proof}

%% file: app-modulus.tex
\section{Extensions to modulus constraints}

Recall that in the body of the paper, cardinality constraints, either within $\GPtwo$ or the global cardinality constraints added to $\Ctwo$, did not constrain the modulus of a count. So, for example, in $\GPtwo$ we could not state that the number of neighbors of an element with a given property is even. And similarly, we could not add to $\Ctwo$ a  constraint that the total number of elements in a model satisfying a formula is an even number. We stated in the body that we can add such modulus constraints and obtain the same results.

In this appendix we will first indicate how to extend the satisfiability result for $\GPtwo$ to the case when existential quantifiers are allowed in the Presburger quantifiers, i.e., when the formula $P(x)$ is of the form:
\begin{equation*}
P_i(x)\ :=\ 
\exists x_1 \cdots \exists x_n\ \sum_{t \in \intsinterval{m}} \lambda_{t} \cdot \presby{R_t(x, y) \land x \neq y}
\ + \ \sum_{t\in \intsinterval{n}} \lambda_t'\cdot x_t
\ \circledast_i\ \delta.
\end{equation*}
Note that the existential variables captures a modulus equation, as $a \equiv b \bmod p$ can be rewritten as:
$$
\exists x_1\exists x_2\
a+x_1p = b+x_2p
$$

We claim that Theorem~\ref{thm:gptwoexp} still holds under such a generalization.
The proof is basically the same.
The only difference is that in the definition of characteristic system
$\cC^\varphi_\pi(\bfx)$
defined in Definition~\ref{def:characteristic-system} the variables $\bfx$ include all the existential variables.
Similarly for the system $\cQ^\varphi(\bfx)$
defined in Definition~\ref{def:linear-system-Q}.
Then, we observe that
Lemmas~\ref{lem:equiv1},~\ref{lem:equiv2}, and~\ref{lem:equiv3} still hold for such systems
and the satisfiability problem of $\GPtwo$ stays in $\exptime$.

For the extension of $\Ctwo$ with global cardinality and modulus constraints, the situation is more obvious. We have already shown how to get a (modulus-constraint free) semilinear representation for the type spectrum vectors. We can compose this with modulus constraints in addition to linear arithmetic constraints and solve using standard techniques for Presburger arithmetic. 
Rewriting modulus constraint into a single linear equation does not change the number of constraints in each disjunct in the formula $\Psi^{\varphi}(\bfx)$.
We emphasize again that the argument here is completely generic, and the same comment holds to show that for \emph{any} decidability extension of Presburger arithmetic (e.g. B\"uchi Arithemetic or Semenov Arithmetic \cite{semenovpresexp}), $\Ctwo$ with global cardinality constraints of that form is decidable.

%% file: main.bbl
\begin{thebibliography}{10}

\bibitem{baadercardinality}
Franz Baader.
\newblock Expressive cardinality restrictions on concepts in a description logic with expressive number restrictions.
\newblock {\em SIGAPP Appl. Comput. Rev.}, 19(3), 2019.

\bibitem{percentage}
Bartosz Bednarczyk, Maja Orlowska, Anna Pacanowska, and Tony Tan.
\newblock {On classical decidable logics extended with percentage quantifiers and arithmetics}.
\newblock In {\em FSTTCS}, 2021.

\bibitem{ultimatelyperiodic}
Michael Benedikt, Egor Kostylev, and Tony Tan.
\newblock {Two Variable Logic with Ultimately Periodic Counting}.
\newblock {\em SIAM Journal on Computing}, 53(4):884--968, 2024.

\bibitem{usicalp}
Michael Benedikt, Chia-Hsuan Liu, Boris Motik, and Tony Tan.
\newblock Verification of graph neural networks via logical characterizations.
\newblock In {\em ICALP}, 2024.

\bibitem{semilinear}
Dmitry Chistikov and Christoph Haase.
\newblock {The Taming of the Semi-Linear Set}.
\newblock In {\em ICALP}, 2016.

\bibitem{caratheodory-integer}
F.~Eisenbrand and G.~Shmonin.
\newblock Carath{\'{e}}odory bounds for integer cones.
\newblock {\em Oper. Res. Lett.}, 34(5):564--568, 2006.

\bibitem{flow}
D.~R. Ford and D.~R. Fulkerson.
\newblock {\em Flows in Networks}.
\newblock Princeton University Press, USA, 2010.

\bibitem{kleenestar}
Christoph Haase and Georg Zetzsche.
\newblock {Presburger arithmetic with stars, rational subsets of graph groups, and nested zero tests}.
\newblock In {\em {LICS}}, 2019.

\bibitem{IP}
Narendra Karmarkar.
\newblock A new polynomial-time algorithm for linear programming.
\newblock In {\em STOC}, 1984.

\bibitem{Kazakov04}
Yevgeny Kazakov.
\newblock A polynomial translation from the two-variable guarded fragment with number restrictions to the guarded fragment.
\newblock In {\em {JELIA}}, 2004.

\bibitem{tonyeryk}
Eryk Kopczynski and Tony Tan.
\newblock Regular graphs and the spectra of two-variable logic with counting.
\newblock {\em {SIAM} J. Comput.}, 44(3):786--818, 2015.

\bibitem{twovarpres}
Chia{-}Hsuan Lu and Tony Tan.
\newblock {On two-variable guarded fragment logic with expressive local Presburger constraints}.
\newblock {\em Logical Methods in Computer Science}, 20, 2024.

\bibitem{papa}
C.~Papadimitriou.
\newblock On the complexity of integer programming.
\newblock {\em J. {ACM}}, 28(4):765--768, 1981.

\bibitem{parikh}
Rohit Parikh.
\newblock {On Context-Free Languages}.
\newblock {\em JACM}, 13(4):570–581, 1966.

\bibitem{Pottier}
Loic Pottier.
\newblock Minimal solutions of linear diophantine systems: Bounds and algorithms.
\newblock In {\em RTA}, 1991.

\bibitem{ctwocomplexity}
Ian Pratt-Hartmann.
\newblock {Complexity of the two-variable fragment with counting quantifiers}.
\newblock {\em {Journal of Logic Language and Information}}, 14:369–395, 2005.

\bibitem{gctwocomplexity}
Ian Pratt-Hartmann.
\newblock Complexity of the guarded two-variable fragment with counting quantifiers.
\newblock {\em Journal of Logic and Computation}, 17(1):133--155, 2007.

\bibitem{ianrevisited}
Ian Pratt{-}Hartmann.
\newblock The two-variable fragment with counting.
\newblock In {\em {WoLLIC}}, 2010.

\bibitem{rudolphcardinality}
"Johann"~Sebastian Rudolph.
\newblock Presburger concept cardinality constraints in very expressive description logics - allegro sexagenarioso ma non ritardando.
\newblock In {\em Description Logic, Theory Combination, and All That - Essays Dedicated to Franz Baader on the Occasion of His 60th Birthday}, 2019.

\bibitem{tree}
Helmut Seidl, Thomas Schwentick, Anca Muscholl, and Peter Habermehl.
\newblock Counting in trees for free.
\newblock In {\em ICALP}, 2004.

\bibitem{semenovpresexp}
Aleksei~L. Semenov.
\newblock Logical theories of one-place functions on the set of natural numbers.
\newblock {\em Math.~USSR Izv.}, 22(3):587--618, 1984.

\bibitem{traktenbrot}
Boris Trakhtenbrot.
\newblock {The Impossibility of an Algorithm for the Decidability Problem on Finite Classes}.
\newblock {\em Proceedings of the USSR Academy of Sciences}, 70(4):{569–572}, 1950.

\end{thebibliography}
